\documentclass[11pt]{article}

\usepackage{latexsym}
\usepackage{amsmath}
\usepackage{amssymb}
\usepackage{amsfonts,amsthm}
\usepackage{enumitem}
\usepackage[top=1in, bottom=1in, left=1in, right=1in]{geometry}

\newcommand{\ignore}[1]{}

\newtheorem{theorem}{Theorem}[section]
\newtheorem{lemma}[theorem]{Lemma}

\newtheorem{corollary}[theorem]{Corollary}

\newtheorem{remark}[theorem]{Remark}

\newtheorem{proposition}[theorem]{Proposition}

\usepackage{thmtools}
\usepackage{thm-restate}

 \usepackage{url}
 \usepackage{subcaption}
\usepackage{crop}
\usepackage{enumerate}
\usepackage{etoolbox}
\usepackage{float}
\usepackage{hyperref}
\usepackage{latexsym}
\usepackage{mathtools}
\usepackage{relsize}
\usepackage{amsmath}
\usepackage{amssymb}
\usepackage{amsfonts}
\usepackage{balance}
\usepackage{tabu}
\usepackage{longtable}

\usepackage{multirow}
\usepackage{lscape}
\usepackage[lined, algoruled, linesnumbered]{algorithm2e}
\usepackage{lipsum}

\let\oldReturn\Return
\renewcommand{\Return}{\State\oldReturn}

\usepackage{abstract}
\usepackage{hyperref,xcolor}
\definecolor{winered}{rgb}{0.5,0,0}

\usepackage{enumitem}

\hypersetup
{
    pdfborder={0 0 0},
    colorlinks=true,
    linkcolor={winered},
    urlcolor={winered},
    filecolor={winered},
    citecolor={winered},
    linktoc=all,
}

\setlist[description]{leftmargin=\parindent,labelindent=\parindent}

\begin{document}

\title{Linear-Time Algorithms for Computing Twinless Strong Articulation Points and Related Problems\thanks{Research supported by the Hellenic Foundation for Research and Innovation (H.F.R.I.) under the ``First Call for H.F.R.I. Research Projects to support Faculty members and Researchers and the procurement of high-cost research equipment grant'', Project FANTA (eFficient Algorithms for NeTwork Analysis), number HFRI-FM17-431.}}
\author{Loukas Georgiadis$^{1}$ \and Evangelos Kosinas$^{2}$}

\maketitle

\begin{abstract}
A directed graph $G=(V,E)$ is twinless strongly connected if it contains a strongly connected spanning subgraph without any pair of antiparallel (or \emph{twin}) edges. The twinless strongly connected components (TSCCs) of a directed graph $G$ are its maximal twinless strongly connected subgraphs. These concepts have several diverse applications, such as the design of telecommunication networks and the structural stability of buildings.
A vertex $v \in V$ is a twinless strong articulation point of $G$, if the deletion of $v$ increases the number of TSCCs of $G$.
%
%
We show that the computation of twinless strong articulation points reduces to the following problem in undirected graphs, which may be of independent interest:
Given a $2$-vertex-connected (biconnected) undirected graph $H$, find all vertices $v$ that belong to a vertex-edge cut-pair, i.e., for which there exists an edge $e$ such that
$H \setminus \{v,e\}$ is not connected.
We develop a linear-time algorithm that not only finds all such vertices $v$, but also computes the number of edges $e$ such that
$H \setminus \{v,e\}$ is not connected. This also implies that for each twinless strong articulation point $v$ which is not a strong articulation point in a strongly connected digraph $G$, we can compute the number of TSCCs in $G \setminus v$.
We note that the problem of computing all vertices that belong to a vertex-edge cut-pair
can be solved in linear-time by exploiting the structure of $3$-vertex-connected (triconnected) components of $H$, represented by an SPQR tree of $H$.
Our approach, however, is conceptually simple, and thus likely to be more amenable to practical implementations.
\end{abstract}

\footnotetext[1]{Department of Computer Science \& Engineering, University of Ioannina, Greece. E-mail: loukas@cs.uoi.gr}
\footnotetext[2]{Department of Computer Science \& Engineering, University of Ioannina, Greece. E-mail: ekosinas@cs.uoi.gr}

\section{Introduction}
\label{sec:intro}

Let $G=(V,E)$ be a directed graph (digraph), with $m$ edges and $n$ vertices. Digraph $G$ is \emph{strongly connected} if there is a directed path from each vertex to every other vertex.
The \emph{strongly connected components} (SCCs) of $G$ are its maximal strongly connected subgraphs. Two vertices $u,v \in V$  are \emph{strongly connected} if they belong to the same strongly connected component of $G$.
We refer to a pair of antiparallel edges, $(x,y)$ and $(y,x)$, of $G$ as \emph{twin edges}.
A digraph $G=(V,E)$ is \emph{twinless strongly connected} if it contains a strongly connected spanning subgraph $(V,E')$ without any pair of twin edges.
The \emph{twinless strongly connected components} (TSCCs) of $G$ are its maximal twinless strongly connected subgraphs. Two vertices $u,v \in V$  are \emph{twinless strongly connected} if they belong to the same twinless strongly connected component of $G$.
Twinless strong connectivity is motivated by several diverse applications, such as the design of telecommunication networks and the structural stability of buildings~\cite{Raghavan2006}.
Raghavan~\cite{Raghavan2006} provided a characterization of twinless strongly connected digraphs, and, based on this characterization, provided a linear-time algorithm for computing the TSCCs of a digraph.

In this paper, we further explore the notion of twinless strong connectivity, with respect to $2$-connectivity in digraphs.
An edge (resp., a vertex) of a digraph $G$ is a \emph{strong bridge} (resp., a \emph{strong articulation point}) if its removal increases the number of strongly connected components. Thus, strong bridges (resp., strong articulation points) are $1$-edge (resp., $1$-vertex) cuts for digraphs.
A strongly connected digraph $G$ is \emph{$2$-edge-connected} if it has no strong bridges, and it is
\emph{$2$-vertex-connected} if it has at least three vertices and no strong articulation points.
Let $C \subseteq V$.
The induced subgraph of $C$, denoted by $G[C]$, is the subgraph of $G$ with vertex set $C$ and edge set $E \cap (C \times C)$.
If $G[C]$ is $2$-edge-connected (resp., $2$-vertex-connected), and there is no set of vertices $C'$ with $C \subsetneq C' \subseteq V$ such that $G[C']$ is also
$2$-edge-connected (resp., $2$-vertex-connected), then $G[C]$ is a \emph{maximal $2$-edge-connected (resp., $2$-vertex-connected) subgraph of $G$}.
Two vertices $u, v\in V$ are said to be \emph{$2$-edge-connected} (resp., \emph{$2$-vertex-connected}) if there are two edge-disjoint  (resp., two internally vertex-disjoint) directed paths from $u$ to $v$ and two edge-disjoint  (resp., two internally vertex-disjoint) directed paths from $v$ to $u$ (note that a path from $u$ to $v$ and a path from $v$ to $u$ need not be edge- or vertex-disjoint). A \emph{$2$-edge-connected component} (resp., \emph{$2$-vertex-connected component}) of a digraph $G=(V,E)$ is defined as a maximal subset $B \subseteq V$ such that every two vertices $u, v \in B$ are $2$-edge-connected (resp., $2$-vertex-connected).
We note that connectivity-related problems for digraphs are known to be much more difficult than for undirected graphs, and indeed many notions for undirected connectivity do not translate to the directed case. See, e.g., \cite{CHILP:2CS,2C:GIP:SODA,2CC:HenzingerKL:ICALP15}.
Indeed, it has only recently been shown that all strong bridges and strong articulation points of a digraph can be computed in linear time~\cite{Italiano2012}.
Additionally, it was shown very recently how to compute the $2$-edge- and $2$-vertex-connected components of digraphs in linear time \cite{2ECC:GILP:TALG,2VCB:IC}, while
the best current bound for computing the maximal $2$-edge- and the $2$-vertex-connected subgraphs in digraphs is not even linear, but it is
$O(\min\{m^{3/2}, n^2\})$~\cite{CHILP:2CS,2CC:HenzingerKL:ICALP15}.

\begin{figure}[t!]
\begin{center}
\centerline{\includegraphics[trim={0 0 0 12.5cm}, clip=true, width=\textwidth]{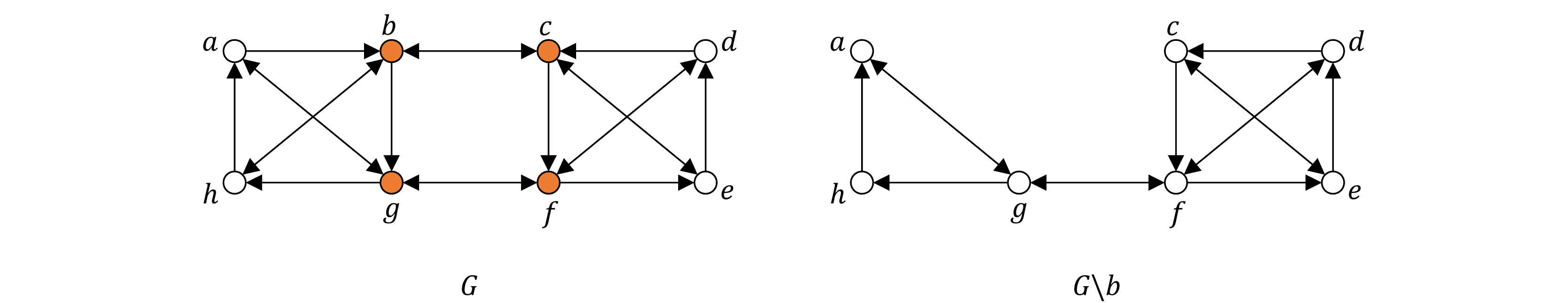}}
\caption{A $2$-vertex-connected digraph $G$ that is not twinless $2$-vertex-connected. Vertices $b$, $c$, $g$ and $f$ are twinless strong articulation points but not strong articulation points; for instance, $G\setminus b$ is strongly connected but not twinless strongly connected.}
\label{figure:TwinlessSAP-example}
\end{center}
\end{figure}

The above notions extend naturally to the case of twinless strong connectivity.
An edge $e \in E$ is a \emph{twinless strong bridge} of $G$ if the deletion of $e$ increases the number of TSCCs of $G$. Similarly,
a vertex $v \in V$ is a \emph{twinless strong articulation point} of $G$ if the deletion of $v$ increases the number of TSCCs of $G$.
A linear-time algorithm for detecting all twinless strong bridges can be derived by combining the linear-time algorithm of Italiano et al.~\cite{Italiano2012} for computing all the strong bridges of a digraph and 
a linear-time algorithm for computing all the edges which belong to a cut-pair in a $2$-edge-connected undirected graph. (See Section \ref{sec:twinless-bridges} for details.)
Previously, Jaberi~\cite{jaberi2019twinless} studied the properties of twinless strong articulation points and some related concepts, and presented an $O(m(n-s))$-time algorithm for their computation, where $s$ is the number of strong articulation points of $G$.
Hence, this bound is $O(mn)$ in the worst case.
Here, we present a linear-time algorithm that identifies all the twinless strong articulation points.
Specifically,  we show that the computation of twinless strong articulation points reduces to the following problem in undirected graphs, which may be of independent interest:
Given a $2$-vertex-connected (biconnected) undirected graph $H$, find all vertices $v$ that belong to a vertex-edge cut-pair, i.e., for which there exists an edge $e$ such that
$H \setminus \{v,e\}$ is not connected.
We develop a linear-time algorithm that not only finds all such vertices $v$, but also computes the number of vertex-edge cut-pairs of $v$ (i.e., the number of edges $e$ such that
$H \setminus \{v,e\}$ is not connected). This implies that for each twinless strong articulation point $v$, that is not a strong articulation point in a digraph $G$, we can also compute the number of twinless strongly connected components of $G \setminus v$.

Our algorithm exploits properties of depth-first search (DFS), and concepts that are defined on the structure given by the DFS,
which are reminiscent of the seminal $3$-vertex-connected (triconnected) components algorithm of Hopcroft and Tarjan~\cite{3-connectivity:ht}.
Indeed, we can compute the vertices that form a vertex-edge cut-pair by exploiting the structure of the triconnected components of $H$, represented by an SPQR tree~\cite{SPQR:triconnected,SPQR:planarity} of $H$. (See Appendix~\ref{section:SPQR} for details.)
In order to construct an SPQR tree, however, we need to know the triconnected components of the graph~\cite{SPQR:GM}, and efficient algorithms that compute triconnected components are considered conceptually complicated (see, e.g., \cite{3connectivity-LOCAL,SPQR:GM,3-connectivity:ht}).
Our approach, on the other hand, is conceptually simple and thus likely to be more amenable to practical implementations.
Also, we believe that our results and techniques will be useful for the design of faster algorithms for related connectivity problems, such as computing twinless $2$-connected components~\cite{jaberi20192edgetwinless,jaberi2019computing}.

\section{Preliminaries}
\label{sec:preliminaries}

Let $G$ be a (directed or undirected) graph. We denote by $V(G)$ and $E(G)$, respectively, the vertex set and edge set of $G$.
For a set of edges (resp., vertices) $S$, we let $G \setminus S$ denote the graph that results from $G$ after deleting the edges in $S$ (resp., the vertices in $S$ and their adjacent edges). We extend this notation for mixed sets $S$, that may contain both vertices and edges of $G$, in the obvious way. Also, if $S$ has only one element $x$, we abbreviate $G \setminus S$ by $G \setminus x$.
Let $C \subseteq V(G)$.
The induced subgraph of $C$, denoted by $G[C]$, is the subgraph of $G$ with vertex set $C$ and edge set $E \cap (C \times C)$.

For any digraph $G$, the associated undirected graph $G^u$ is the graph with vertices $V(G^u)=V(G)$ and edges $E(G^u) = \{ \{u,v\} \mid (u,v) \in E(G) \vee (v,u) \in E(G) \}$.
Let $H$ be an undirected graph. An edge $e \in E(H)$ is a \emph{bridge} if its removal increases the number of connected components of $H$. A connected graph $H$ is $2$-edge-connected if it contains no bridges.
Raghavan~\cite{Raghavan2006} proved the following characterization of twinless strongly connected digraphs.

\begin{theorem}(\cite{Raghavan2006})
\label{theorem:TwinlessCharacterization}
Let $G$ be a strongly connected digraph. Then $G$ is twinless strongly connected if and only if its underlying undirected graph $G^u$ is $2$-edge-connected.
\end{theorem}

Theorem~\ref{theorem:TwinlessCharacterization} implies a linear-time algorithm to compute the twinless strongly connected components (TSCCs) of a digraph $G$. It suffices to compute the strongly connected components, $C_1, \ldots, C_k$, of $G$ and compute the $2$-edge-connected components of each underlying undirected graph $G^{u}[C_i]$, $1 \le i \le k$. All these computations take linear time~\cite{dfs:t}.

\section{Computing twinless strong bridges}
\label{sec:twinless-bridges}

Another immediate consequence of Theorem~\ref{theorem:TwinlessCharacterization} is that a twinless strong bridge in a twinless strongly connected graph is either $(1)$ a strong bridge or $(2)$ an edge whose removal destroys the $2$-edge connectivity in the underlying graph. All strong bridges can be found in linear time~\cite{Italiano2012}. To compute the edges of type $(2)$, we only have to find all the edges of the underlying graph whose removal destroys the $2$-edge connectivity.

\begin{lemma}
\label{lemma:tsb}
{{}}
An edge $e=(x,y)$ in a twinless strongly connected digraph $G$ is a twinless strong bridge but not a strong bridge if and only if its twin $(y,x)$ is not an edge of $G$ and $(G \setminus e)^u$ is not $2$-edge-connected.
\end{lemma}
\begin{proof}
Let $e=(x,y)$ be an edge in a twinless strongly connected digraph $G$ which is a twinless strong bridge but not a strong bridge. Theorem~\ref{theorem:TwinlessCharacterization} implies, that the removal of this edge leaves us with a graph $H=G \setminus e$ whose underlying undirected graph $H^u$ is not $2$-edge-connected. Now, since the initial graph is twinless strongly connected, its underlying undirected graph $G^u$ is $2$-edge-connected; therefore the twin $(y,x)$ of $e$ is not an edge of $G$ (otherwise, the removal of $e$ would leave the underlying graph unchanged). The converse is an immediate consequence of Theorem~\ref{theorem:TwinlessCharacterization}.
\end{proof}

This suggests the following algorithm for the computation of all twinless strong bridges. First, we mark all edges that are strong bridges. Then we find all edges in the underlying graph which belong to a cut-pair, and mark those that correspond to an edge in the initial graph whose twin is missing from it. An algorithm which computes, in linear time, all the edges which belong to a cut-pair in a $2$-edge-connected undirected graph is given by Tsin~\cite{TSIN:3EC}.

In Appendix \ref{section:finding_cut_pairs} we apply the framework that we develop in Section \ref{sec:dfs}, in order to provide an alternative linear-time algorithm for computing all the edges that belong to a cut-pair in a $2$-edge-connected undirected graph.
Our algorithm also counts, for every edge $e$, all edges $e'$ such that $\{e,e'\}$ is a cut-pair. This is useful for counting the TSCCs after the removal of a twinless strong bridge that is not a strong bridge, in a twinless strongly connected digraph (see Lemma \ref{lemma:counting_tsccs_edge}).
Furthermore, we describe how a minor extension of our algorithm can efficiently answer queries of the form ``report all edges that form a cut-pair with $e$''.

\section{Computing twinless strong articulation points}
\label{sec:twinless-sap}

It is an immediate consequence of Theorem~\ref{theorem:TwinlessCharacterization} that a twinless strong articulation point in a twinless strongly connected digraph $G$ is either $(1)$ a strong articulation point or $(2)$ a vertex whose removal destroys the $2$-edge connectivity in the underlying undirected graph $G^u$.
Since all strong articulation points can be computed in linear time~\cite{Italiano2012}, it remains to find all vertices of type $(2)$.
Note that such a vertex $x$ either $(a)$ entirely destroys the connectivity of the underlying graph $G^{u}$ with its removal, or $(b)$, upon removal, it leaves us with a graph $G^{u} \setminus x$ that is connected but not $2$-edge-connected.
Clearly, the set of vertices with property $(a)$ are a subset of the set of strong articulation points.
Therefore, it suffices to find all vertices with property $(b)$.
To that end, we process each $2$-vertex-connected component of $G^{u}$ separately, as the next lemma suggests.

\begin{lemma}
\label{lemma:2vcc}
Let $H$ be a $2$-edge-connected undirected graph. Let $v$ be a vertex that is not an articulation point, and let $C$ be its $2$-vertex-connected component (2VCC). For any edge $e$, $H \setminus \{v,e\}$ is not connected if and only if $e$ belongs to $C$ and $C \setminus \{v,e\}$ is not connected.
\end{lemma}
\begin{proof}
\label{proof:vcc}
($\Rightarrow$) Since $H$ is $2$-edge-connected, it contains no bridges. Therefore, every 2VCC of $H$ contains at least three vertices, and thus it is a $2$-vertex-connected subgraph of $G$. Now, let $e$ be an edge such that the graph $H \setminus \{v,e\}$ is not connected, and let $C'$ be the 2VCC of $H$ that contains $e$. 
Suppose, for contradiction, that $C' \ne C$. Since $C'$ is $2$-vertex-connected, is must also be $2$-edge-connected, so $e$ is not a bridge in $C'$.
Moreover, $v$ is not an articulation point, so it is contained in only one 2VCC of $H$.
Hence, $C' \setminus \{v\}=C'$. This means that $e$ is not a bridge in $C' \setminus \{v\}$, and therefore it is not a bridge in $H\setminus\{v\}$.\\
($\Leftarrow$) 
Recall the block graph representation of $H$~\cite{graph_theory:diestel}:
Let $T$ be the graph whose vertices are the 2VCCs and the articulation points of $H$, and which contains an edge $e$ if and only if $e$ connects a 2VCC $C'$ with an articulation point $x \in C'$; then $T$ is a tree. Now, suppose that there exists an edge $e=(x,y)$ in $C$ such that $C \setminus \{v,e\}$ is not connected, but $G \setminus \{v,e\}$ is connected. This means that there exists a simple path $P$ in $G \setminus \{v,e\}$ connecting $x$ and $y$. Since $x$ and $y$ are not connected in $C \setminus \{v,e\}$, $P$ must contain vertices from $G \setminus C$. So let $z$ be the first vertex in $P$ such that $z \in C$ but its successor in $P$ is not (such a vertex exists, since $x \in C$). Since $z$ is in $C$ and has a neighbor that belongs to a different 2VCC, it is an articulation point. Now let $w$ be the first vertex after $z$ in $P$ such that $w \in C$ (such a vertex exists, since $y \in C$). Due to the tree structure of the 2VCCs of $H$, we conclude that $w=z$. (In other words, when a path leaves a 2VCC through an articulation point, in order to return to this 2VCC it must pass again through the same articulation point.) But this contradicts the simplicity of $P$.
\end{proof}

So, in order to find all twinless strong articulation points, it is sufficient to solve the following problem:
\emph{Given a $2$-vertex-connected undirected graph $G$, find all vertices $v$ for which there exists an edge $e$ such that 
$G \setminus \{v,e\}$ is not connected.}
Next, we describe a linear-time algorithm for this problem.
Our algorithm utilizes properties of depth-first search (DFS), which are reminiscent of the seminal algorithm of Hopcroft and Tarjan for computing the triconnected components of a graph~\cite{3-connectivity:ht}.

Formally, our main technical contribution is summarized in the following theorem:
\begin{theorem}
\label{theorem:cut-pairs}
Let $G$ be an undirected graph. There is a linear time algorithm that computes, for all vertices $v$ that are not articulation points of $G$, the number of edges $e$ such that $G \setminus \{v,e\}$ is not connected.
\end{theorem}

Then, Theorem~\ref{theorem:cut-pairs} implies the following results:

\begin{corollary}
\label{corollary:twinless-SAP}
Let $G$ be a twinless strongly connected digraph. There is a linear time algorithm that finds all the twinless strong articulation points of $G$. Moreover, for all twinless strong articulation points $v$ that are not strong articulation points of $G$, the algorithm computes (in linear time) the number of TSCCs in $G \setminus v$.
\end{corollary}

\begin{corollary}
\label{corollary:cut-pairs}
Let $G$ be a biconnected undirected graph. After linear-time preprocessing, we can answer queries of the form: Given a vertex $v$ of $G$, report all the edges in the set $C(v) = \{ e \in E(G)  \mid  G \setminus \{v, e\}$ {is not connected} $\}$, in $O(|C(v)|)$ time.
\end{corollary}

\subsection{Depth-first search, low and high points}
\label{sec:dfs}

Let $G$ be a $2$-vertex-connected graph.
We consider a DFS traversal of $G$, starting from an arbitrarily selected vertex $r$, and let $T$ be the resulting DFS tree~\cite{dfs:t}.
A vertex $u$ is an ancestor of a vertex $v$ ($v$ is a descendant of $u$) if the tree path from $r$ to $v$ contains $u$.
Thus, we consider a vertex to be an ancestor (and, consequently, a descendant) of itself.
We let $p(v)$ denote the parent of a vertex $v$ in $T$.
If $u$ is a descendant of $v$ in $T$, we denote the set of vertices of the simple tree path from $u$ to $v$ as $\mathit{T}[u,v]$. The expressions $\mathit{T}[u,v)$ and $\mathit{T}(u,v]$ have the obvious meaning (i.e., the vertex on the side of the parenthesis is excluded from the tree path).
Furthermore, we let $T(v)$ denote the subtree of $T$ rooted at vertex $v$.
We identify vertices in $G$ by their DFS number, i.e., the order in which they were discovered by the search.
Hence, $u \leq v$ means that vertex $u$ was discovered before $v$.
The edges in $E(T)$ are called tree-edges; the edges in $E(G) \setminus E(T)$ are called back-edges, as their endpoints have ancestor-descendant relation in $T$.
When we write $(v,w)$ to denote a back-edge, we always mean that $w \le v$, i.e., $v$ is an descendant of $w$ in $T$.

\begin{figure}[t!]
\begin{center}
\centerline{\includegraphics[trim={0 0 0 4.5cm}, clip=true, width=\textwidth]{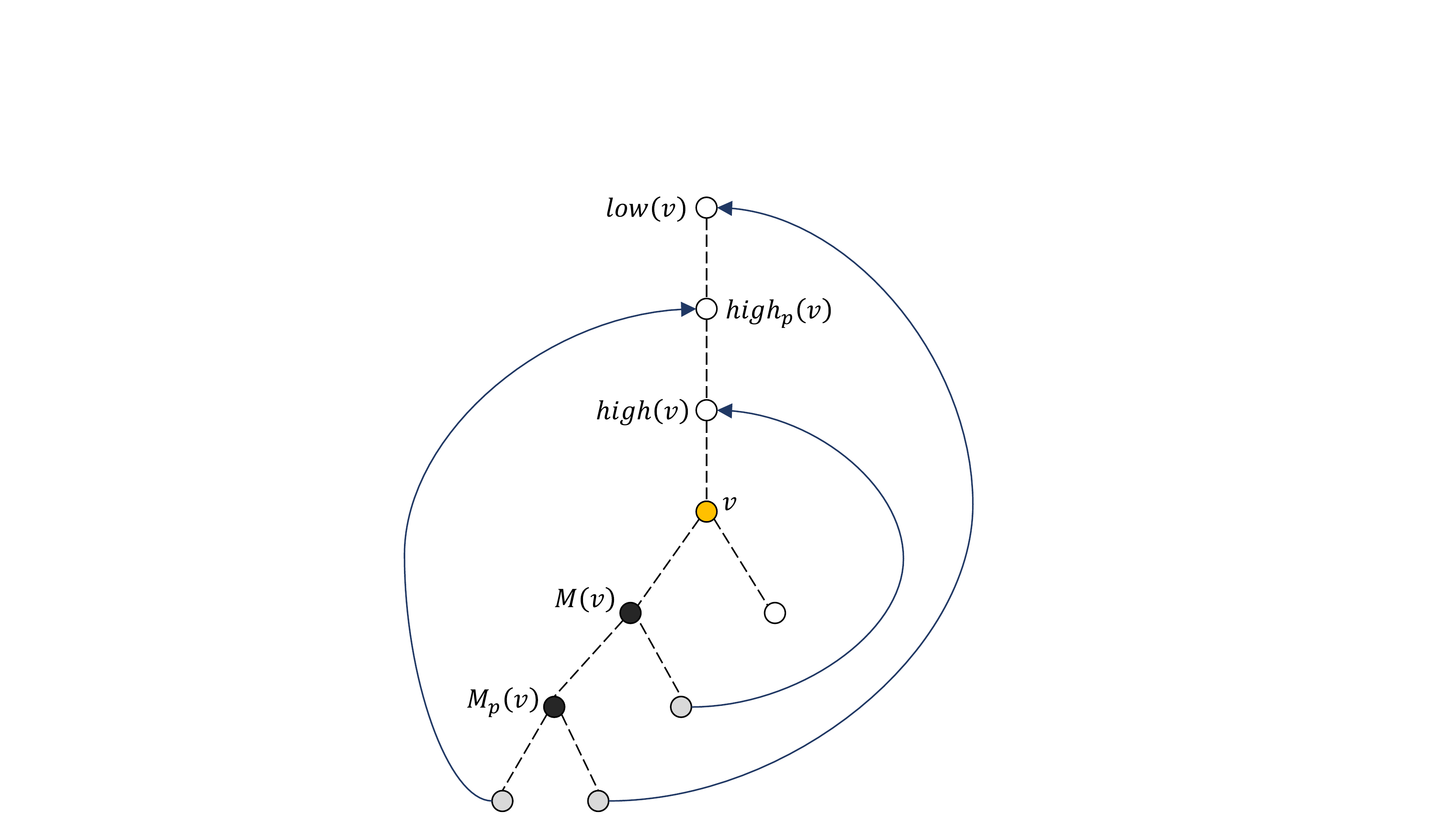}}
\caption{Concepts defined on the structure of the DFS tree that are essential to our algorithm.
Dashed lines correspond to DFS tree paths. Back-edges are shown directed from descendant to ancestor.}
\label{figure:DFS}
\end{center}
\end{figure}

Now we describe some concepts that are defined on the structure given by the DFS and are essential to our algorithm. For an illustration, see Figure~\ref{figure:DFS}. Let us note beforehand, that all these concepts are well-defined, for $G$ is $2$-vertex-connected.
Define the “low point”, $\mathit{low}(v)$, of a vertex $v\neq r$, as the minimum vertex (w.r.t. the DFS numbering) that is connected via a back-edge to a descendant of $v$, i.e.,
the minimum vertex in the set $\{u \mid $ {there is a back-edge} $ (w,u) $ {such that} $w$ {is a descendant of} $v\}$.
Define the “high point”, $\mathit{high}(v)$, of $v\neq r$, as the maximum proper ancestor of $v$
which is connected with a back-edge to a descendant of $v$.
The notion of low points plays central role in classic algorithms for computing the biconnected components~\cite{dfs:t} and the triconnected components~\cite{3-connectivity:ht} of a graph. The low points of all vertices can be computed in linear time.
Hopcroft and Tarjan~\cite{3-connectivity:ht} also use a concept of high points, which, however, is different from ours.
Since $G$ is $2$-vertex-connected, $r$ has a unique child vertex $c$ in $T$.
For any vertex $v \not= r,c$, define $\mathit{high_p}(v)$ to be the maximum proper ancestor of $\mathit{p}(v)$ which is connected with a back-edge to a descendant of $v$.
Finally, for any vertex $v\neq r$, define $\mathit{M}(v)$ to be the nearest common ancestor
of all descendants of $v$ that are connected with a back-edge to a proper ancestor of $v$.
For any vertex $v\neq r, c$, define $\mathit{M_p}(v)$ as the maximum common ancestor of all descendants of $v$ that are connected with a back-edge to a proper ancestor of $\mathit{p}(v)$.
Next, we show how to compute $\mathit{high}(v)$, $\mathit{high_p}(v)$, $M(v)$ and $M_p(v)$, for all vertices $v$, in total linear time.

\subsection{Finding all $\mathit{high}(v)$ and $\mathit{high_p}(v)$ in linear time}
\label{sec:high}

The basic idea to compute all $\mathit{high}(v)$ (for $v\ne r$) is to do the following:
We process the back-edges $(u,v)$ in decreasing order with respect to their lower end $v$.
When we process $(u,v)$, we ascend the path $T[u,v]$, and for each visited vertex $x$ such that
$\mathit{high}[x]$ is still undefined, we set
$\mathit{high}[x] \leftarrow v$. (See Algorithm~\ref{algorithm:high}.)
It should be clear that this process, which forms the basis of our linear-time algorithm, computes all $\mathit{high}(v)$, for $v\ne r$, correctly.

\SetKwComment{tcc}{/*}{*/}

\begin{algorithm}[!h]
\caption{\textsf{High}}
\label{algorithm:high}
\LinesNumbered
\DontPrintSemicolon

\lForEach{vertex $v \not= r$}{set $\mathit{high}[v] \leftarrow \mathit{null}$}
sort the back-edges $(u,v)$ in decreasing order w.r.t. to their lower end $v$\;

\ForEach{back-edge $(u,v)$}{
\While{$u>v$}{
\lIf{$\mathit{high}[u] = \mathit{null}$}{$\mathit{high}[u] \leftarrow v$}
$u \leftarrow \mathit{p}(u)$
}
}
\end{algorithm}

In order to achieve linear running time, we have to be able, when we consider a back-edge $(u,v)$, to bypass all vertices on the path $\mathit{T}[u,v]$ whose $\mathit{high}$ value has been computed.
To that end, it suffices to know, for every vertex $x$ in $\mathit{T}[u,v]$, the nearest ancestor of $x$ whose $\mathit{high}$ value is still null.
We can achieve this by applying a disjoint-set-union (DSU) structure~\cite{dsu:tarjan}.
Specifically, we maintain a forest $F$ that is a subgraph of $T$, subject to the following operations:
\begin{description}
\item[$\mathit{link}(x,y)$:] Adds the edge $(x,y)$ into the forest $F$.
\item[$\mathit{find}(x)$:] Return the root of the tree in $F$ that contains $x$.
\end{description}
Let $F_x$ denote the tree of $F$ that contains a vertex $x$.
Initially, $F$ contains no edges, so $x$ is the unique vertex in $F_x$.
In our algorithm, the link operation always adds some tree edge $(u,p(u))$ to $F$, so the invariant that $F$ is a subgraph of $T$ is maintained. This is implemented by uniting the corresponding sets of $u$ and $p(u)$ in the underlying DSU structure, and setting the root of $F_{p(u)}$ as the representative of the resulting set.
Then, $\mathit{find}(u)$ returns the root of $F_u$, which will be the nearest ancestor of $u$ in $T$ whose $\mathit{high}$ value is still null.
Algorithm~\ref{algorithm:fasthigh} gives a fast algorithm for computing
$\mathit{high}(v)$, for every vertex $v\ne r$.

\begin{algorithm}[!h]
\caption{\textsf{FastHigh}}
\label{algorithm:fasthigh}
\LinesNumbered
\DontPrintSemicolon

initialize a forest $F$ with $V(F)=V(T)$ and $E(F) = \emptyset$\;
\lForEach{vertex $v \not= r$}{set $\mathit{high}[v] \leftarrow \mathit{null}$}

sort the back-edges $(u,v)$ in decreasing order w.r.t. to their lower end $v$\;

\ForEach{back-edge $(u,v)$}{
$u \leftarrow \mathit{find}(u)$\;
\While{$u>v$}{\label{highwhile}
$\mathit{high}[u] \leftarrow v$\;
$next \leftarrow \mathit{find}(p(u))$\;
\label{line_in_highFast}
$\mathit{link}(u,p(u))$\;
$u \leftarrow next$
}
}
\end{algorithm}

The next lemma summarizes the properties of Algorithm~\ref{algorithm:fasthigh}.

\begin{lemma}
Algorithm \ref{algorithm:fasthigh} is correct. Furthermore, it will perform $n-1$ \textbf{link} and $2m-n+1$ \textbf{find} operations on a $2$-vertex-connected graph with $n$ vertices and $m$ edges.
\end{lemma}
\begin{proof}
\label{proof:fasthigh}
Let $B$ be the sorted list of the back-edges. (Notice that $B$ contains $m-n+1$ edges.) We will prove the theorem inductively by showing that, for every $t$ in $\{0,\dotsc,m-n\}$: \textbf{if}, after having run the algorithm for the first $t$ back-edges, we now have that, (1) for every vertex $x$, $\mathit{find}(x)$ returns the nearest ancestor of $x$ whose $hig$h value is still null, (2) for every back-edge $(u,v)$ in $B[1,t]$, $\mathit{high}[x]$ has been computed correctly for every $x$ in $T[u,v)$, and the $\mathit{high}$ value of every other vertex, which does not belong to such a set, is still null, and (3) every set that has been formed due to the $\mathit{link}$ operations that have been performed is a subtree of $T$, of whose members only its root has its $high$ value still set to null, \textbf{then}, if we run the algorithm once more for the $t+1$ back-edge, properties (1), (2) (for $t+1$), and (3) will still hold true.

For the basis of our induction, let us note that the premise of the inductive proposition for $t=0$ is trivially true: Before we have begun traversing $B$, the set containing $x$ is a singleton, $\mathit{find}(x)=x$, and $\mathit{high}[x]$ is null, for every vertex $x$. Now, suppose the premise of the inductive proposition is true for some $t$ in $\{0,\dotsc,m-n\}$, and let $(u,v)$ be the $t+1$ back-edge. Let $x_1,\dotsc,x_k$, in decreasing order, be the vertices in $T[u,v]$ whose $\mathit{high}$ value is still null. (Note that, since $B$ is sorted in decreasing order w.r.t the lower end-point of its elements, we have $x_k=v$.) We observe two facts. First, by (1), we have that $x_1=\mathit{find}(u)$, and $x_i=\mathit{find}(p(x_{i-1}))$, for $i=2,\dotsc,k$. Second, (2) implies that the correct $\mathit{high}$ value of $x_i$, for every $i=1,\dotsc,k-1$, is $v$ (although now it is still set to null). From these two facts we can see that, in order to prove that our algorithm is going to correctly compute the values $\mathit{high}[x_i]$, for $i=1,\dotsc,k-1$, and not mess with those that have already been computed, it is sufficient to show that the function $\mathit{find}(p(x))$, in line \ref{line_in_highFast}, will return, every time it is invoked, the closest ancestor of $p(x)$ whose $\mathit{high}$ value is still set to null - despite all the $\mathit{link}$ operations which might have been performed in the meantime. To see this, observe that (1) and (3) imply that, for every $i=1,\dotsc,k-1$, $x_i$ and $p(x_i)$ belong to different sets (since $x_i$, having its $\mathit{high}$ value still set to null, is the root of the set it belongs to). From this we conclude, that after linking $x_i$ with $p(x_i)$, $x_j$ and $p(x_j)$ still belong to different sets, for every $j=i+1,\dotsc,k-1$. It should be clear now that, by executing our algorithm for the $t+1$ back-edge, only the $\mathit{high}$ values of $x_1,\dotsc,x_{k-1}$ are going to be affected (and computed correctly). This shows that (2) (for $t+1$) still holds true. We also see that all the sets that have been formed due to the $\mathit{link}$ operations that have been performed are still subtrees of $T$, since every such operation is linking a vertex with its parent. Now, let $x$ be a vertex that belongs to one of the sets that have been affected by the $\mathit{link}$ operations that have been performed during the execution of the algorithm for the $t+1$ back-edge. By (1), this means that, before running the algorithm for this back-edge, $\mathit{find}(x)=x_i$, for some $i=1,\dotsc,k$. We conclude, that after running the algorithm for the $t+1$ back-edge, the closest ancestor of $x$ that has its $\mathit{high}$ value still set to null is $v$ (since, now, every vertex in $T[u,v)$ has its $\mathit{high}$ value computed). This shows that (1) still holds true. Furthermore, this also shows that every vertex in $T[x,v]$ is part of the same set. We conclude that the root of the set which contains $x$ is $v$. Thus we have shown that (3) still holds true. (We do not have to consider the vertices whose set has not been affected by the $\mathit{link}$ operations.)

Thus we have proved that, since the premise of the inductive proposition for $t=0$ is true, (2) in the conclusion of the inductive proposition for $t=m-n$ is also true. In other words, our algorithm computes correctly the $\mathit{high}$ value of every $x$ which belongs to a set of the form $T[u,v)$, for some back-edge $(u,v)$. Since the graph is $2$-vertex-connected, every vertex $x \ne r$ belongs to such a set. Furthermore, after the execution of the algorithm, precisely $n-1$ $\mathit{link}$ operations (one for every vertex $x \ne r$), and $2m-n-1$ $\mathit{find}$ operations (one for every end-point of every back-edge, and one for every vertex $x \ne r$) will have been performed.
\end{proof}

Since all the $\mathit{link}$ operations we perform are of the type $\mathit{link}(u,p(u))$, and the total number of $\mathit{link}$ and $\mathit{find}$ operations performed is $O(m+n)$, we may use the static tree DSU data structure of Gabow and Tarjan~\cite{dsu:gt} to achieve linear running time.

Finally, we note that the algorithm for computing all $\mathit{high_p}(v)$ is almost identical to Algorithm~\ref{algorithm:fasthigh}. The only difference is in line~\ref{highwhile}, where we have to replace ``\textbf{while} $u>v$'' with ``\textbf{while} $p(u)>v$''. The proof of correctness and linearity is essentially the same.

\subsection{Finding all $\mathit{M}(v)$ and $\mathit{M_p}(v)$ in linear time}
\label{sec:M}

Recall that $\mathit{M}(v)$ is the nearest common ancestor
of all descendants of $v$ that are connected with a back-edge to a proper ancestor of $v$, while
$\mathit{M_p}(v)$ is the nearest common ancestor of all descendants of $v$ that are connected with a back-edge to a proper ancestor of $\mathit{p}(v)$.

Before we describe our algorithm for the computation of $\mathit{M}(v)$ (and $\mathit{M_p}(v)$), we state a lemma that will be useful in what follows.

\begin{lemma}
\label{MvLemma}
Let $u$ and $v$ be such that $v$ is an ancestor of $u$ and $\mathit{M}(v)$ is a descendant of $u$. Then $\mathit{M}(v)$ is a descendant of $\mathit{M}(u)$.
\end{lemma}
\begin{proof}
Let $e=(x,y)$ be a back-edge with $x$ a descendant of $v$ and $y$ a proper ancestor of $v$. Since $v$ is an ancestor of $u$, $y$ is a proper ancestor of $u$. And since $\mathit{M}(v)$ is a descendant of $u$, $x$ is a descendant of $u$. This shows that $\mathit{M}(u)$ is an ancestor of $x$. Since $e$ was chosen arbitrarily, we conclude that $\mathit{M}(u)$ is an ancestor of $\mathit{M}(v)$.
\end{proof}

\begin{remark}
We note that the lemma still holds if we replace $\mathit{M}(v)$ with $\mathit{M_p}(v)$.
\end{remark}

Our algorithm for the computation of $\mathit{M}(v)$ works recursively on the children of $v$. So, let $v$ be a vertex (different from $r$). We define $\mathit{l}(v)= \mathit{min}\{\{u \mid $ there exists a back-edge $(v,u) \}\cup\{v\}\}$. Now, if $\mathit{l}(v)<v$, we have $\mathit{M}(v)=v$. (This is always the case when $v$ is a leaf, since the graph is $2$-vertex-connected.) Furthermore, if there exist two children $c$, $c'$ of $v$ such that $\mathit{low}(c)<v$ and $\mathit{low}(c')<v$, then, again, $\mathit{M}(v)=v$. The difficulty arises when there is only one child $c$ of $v$ with the property $\mathit{low}(c)<v$ (one such child of $v$ must of necessity exist, since the graph is $2$-vertex-connected), in which case $\mathit{M}(v)$ is a descendant of $c$, and, therefore, by Lemma~\ref{MvLemma}, $\mathit{M}(v)$ is a descendant of $\mathit{M}(c)$. In this case, we repeat the same process in $\mathit{M}(c)$: we test whether $\mathit{l}(\mathit{M}(c))<v$ or whether there exists only one child $d$ of $\mathit{M}(c)$ such that $\mathit{low}(d)<v$, in which case we repeat the same process in $\mathit{M}(d)$, and so on.

Now, we claim that a careful implementation of the above procedure yields a linear-time algorithm for the computation of $\mathit{M}(v)$, for all vertices $v\ne r$. To that end, it suffices to store, for every vertex $v$ that is not a leaf of $T$, two pointers, $\mathit{L}(v)$ and $\mathit{R}(v)$, on the list of the children of $v$. Initially, $\mathit{L}(v)$ points to the first child $c$ of $v$ that has $\mathit{low}(c)<v$, and $\mathit{R}(v)$ points to the last child $c'$ of $v$ that has $\mathit{low}(c')<v$. Our algorithm works in a bottom-up fashion. Provided we have computed $\mathit{M}(u)$ for every descendant $u$ of $v$, we execute Procedure~\ref{proc:findM}$(v)$.

\ignore{
\begin{algorithm}[!h]
\caption{\textsf{Mv}}
\label{algorithm:Mv}
\LinesNumbered
\DontPrintSemicolon
\lIf{$l(v)<v$}{\KwRet $v$}
\lIf{$L[v] \ne R[v]$}{\KwRet $v$}
$m \leftarrow M[L[v]]$\;
\While{$M(v) = \emptyset$}{
  \lIf{$l(m)<v$}{\KwRet $m$}
  \lWhile{$\mathit{low}(L[m])\geq v$}{$L[m] \leftarrow$ next child of $m$}
  \label{alg_Mv_line_1}
  \lWhile{$\mathit{low}(R[m])\geq v$}{$R[m] \leftarrow$ previous child of $m$}
  \label{alg_Mv_line_2}
  \lIf{$L[m] \ne R[m]$}{\KwRet $m$}
  $m \leftarrow M[L[m]]$\;
}
\end{algorithm}
}

\begin{procedure}[!h]
\LinesNumbered
\DontPrintSemicolon

	\caption{FindM($v$)}
	\label{proc:findM}

\lIf{$l(v)<v$}{\KwRet $v$}
\lIf{$L[v] \ne R[v]$}{\KwRet $v$}
$m \leftarrow M[L[v]]$\;
\While{$M[v] = \emptyset$}{
  \lIf{$l(m)<v$}{\KwRet $m$}
  \lWhile{$\mathit{low}(L[m])\geq v$}{$L[m] \leftarrow$ next child of $m$}
  \label{alg_Mv_line_1}
  \lWhile{$\mathit{low}(R[m])\geq v$}{$R[m] \leftarrow$ previous child of $m$}
  \label{alg_Mv_line_2}
  \lIf{$L[m] \ne R[m]$}{\KwRet $m$}
  $m \leftarrow M[L[m]]$\;
}
\end{procedure}

\begin{lemma}
\label{lemma:findM}
By executing Procedure~\ref{proc:findM}$(v)$, for all vertices $v \not= r$,  in bottom-up fashion of $T$,
we can compute all $\mathit{M}(v)$ in linear-time.
\end{lemma}
\begin{proof}
\label{proof:findMlemma}
To prove correctness, it is sufficient show that, for the computation of $\mathit{M}(v)$, if $\mathit{M}(v)$ lies in $\mathit{T}(m)$, for some descendant $m$ of $v$, and $\mathit{l}(m)\geq v$, then every back-edge that starts from $\mathit{T}(v)$ and ends in a proper ancestor of $v$ has its starting-point in a subtree of the form $\mathit{T}(c)$, where $c$ is a child of $m$ between $\mathit{L}[m]$ and $\mathit{R}[m]$. It’s easy to see this inductively: that is, let $v$ be a vertex, all of whose descendants had this property as the algorithm was running. Now, suppose that $m$ is a descendant of $v$ such that $\mathit{M}(v)$ lies in $\mathit{T}(m)$. If $\mathit{L}[m]$ points to the first child of $m$ and $\mathit{R}[m]$ to the last child of $m$, then there is nothing to prove. But if one of these two pointers was moved (in lines \ref{alg_Mv_line_1} or \ref{alg_Mv_line_2}) during the execution of the algorithm, this means, thanks to the inductive hypothesis, that for an ancestor $x$ of $m$ which is also a proper descendant of $v$ it is true that every back-edge that starts from $\mathit{T}(x)$ and ends in a proper ancestor of $x$ has its starting-point in a subtree of the form $\mathit{T}(c)$, for some child $c$ of $m$ between $\mathit{L}[m]$ and $\mathit{R}[m]$. Now we see why every back-edge that starts from $\mathit{T}(v)$ and ends in a proper ancestor of $v$ has its starting-point in a subtree of the same form: for if this is not the case, then there exists a back-edge that starts from $\mathit{T}(d)$, for some child $d$ of $m$ which is not between $\mathit{L}[m]$ and $\mathit{R}[m]$, and ends in a proper ancestor of $v$, and this is also a back-edge that starts from $\mathit{T}(x)$ and ends in a proper ancestor of $x$ - a contradiction.

Now, to prove linearity, we note that the only way our algorithm could be making an excessive amount of steps, would be by visiting some vertices a lot of times, when it recursively descends to the descendants of some vertices, in order to compute their $M$ value. So we define, for every vertex $v$, the (possibly empty) list $\mathit{S}(v)=\{m_1,\dotsc,m_{k_v}\}$ of the proper descendants of $v$ to which the algorithm had to descend in order to compute $\mathit{M}(v)$, sorted by the order of visit (i.e. ordered increasingly). (Notice that $m_1=M(c)$, for a child $c$ of $v$, $m_2=M(d)$, for a child $d$ of $m_1$, and so on, and $M(v)=m_{k_v}$.) We will prove linearity by showing that two such distinct lists can meet only in their last element.  Equivalently, we may show that a non-last member $m$ of such a list (let us call it: an intermediary member), can appear only in that list. So, let $m$ be an intermediary member of a list $\mathit{S}(v)$, for some vertex $v$, and let $v$ be the first vertex in whose list $m$ appears as an intermediary member (that is, there is no proper descendant of $v$ in whose list $m$ appears as an intermediary member). We note, that, since $m$ is an intermediary member of $\mathit{S}(v)$, $\mathit{M}(v)$ is a proper descendant of $m$. Now, suppose that there exists a proper ancestor $u$ of $v$ such that $m$ is a member of $\mathit{S}(u)$, and let $u$ be the closest proper ancestor of $v$ that has this property. Then we have $\mathit{l}(u)=u$, and there is a unique child $c$ of $u$ with the property $\mathit{low}(c)<u$. Furthermore, $\mathit{M}(c)$ (the first member of $\mathit{S}(u)$) belongs to $T[m,c]$. But $\mathit{M}(c)$ does not belong to $T[m,v]$: for otherwise, since $c$ is an ancestor of $v$, Lemma \ref{MvLemma} implies that $\mathit{M}(c)$ is a descendant of $\mathit{M}(v)$, which is a proper descendant of $m$. We conclude, that $\mathit{M}(c)$, the first member of $\mathit{S}(u)$, is a proper ancestor of $v$. Now, continuing in this fashion, (i.e. considering the unique child $d$ of $\mathit{M}(c)$ that has the property $\mathit{low}(d)<u$, so that $\mathit{M}(d)$ is the second member of $\mathit{S}(u)$), we see that the members of $\mathit{S}(u)$ are either proper ancestors of $v$ or descendants of $\mathit{M}(v)$, and so none of them can be $m$. A contradiction.
\end{proof}

We use a similar algorithm in order to compute all $\mathit{M_p}(v)$.
The only change we have to make in Procedure~\ref{proc:findM} is to replace every comparison to $v$ with a comparison to $\mathit{p}(v)$. The proof of correctness and linearity is essentially the same.

\section{Finding all vertices that belong to a vertex-edge cut-pair}
Let $H=(V,E)$ be a $2$-vertex-connected undirected graph. For every $v$ in $V$, we define $\mathit{count}(v)$ $:=$ $\#\{e \in E \mid \{v,e\}$ is a cut-pair$\}$. We will find all vertices which belong to a vertex-edge cut-pair of $H$ by computing all $\mathit{count}(v)$. We notice that the parameter $\mathit{count}(v)$ is also useful for counting TSCC, as Lemma~\ref{lemma:counting_tsccs} suggests.
Thus, Corollary~\ref{corollary:twinless-SAP} follows.

\begin{lemma}
\label{lemma:counting_tsccs}
Let $G$ be a twinless strongly connected digraph, and let $v$ be a twinless strong articulation point of $G$ which is not a strong articulation point. Then $\mathit{count}(v)+1$ (computed in the $2$-vertex-connected component of $v$ in $G^u$) is the number of twinless connected components of $G\setminus{v}$.
\end{lemma}
\begin{proof}
Since $v$ is not a strong articulation point, $G\setminus{v}$ is strongly connected. It has been proved in \cite{Raghavan2006}, that the twinless strongly connected components of a strongly connected digraph correspond to the $2$-edge-connected components of its underlying graph. Now, the number of the $2$-edge-connected components of $G^u\setminus{v}$ equals the number of its bridges + $1$ (this is due to the tree structure of the $2$-edge-connected components of a connected graph). By Lemma \ref{lemma:2vcc}, all these bridges lie in the $2$-vertex-connected component of $v$ in $G^u$. By definition, their number is $\mathit{count}(v)$.
\end{proof}

Now, to compute all $\mathit{count}(v)$, we will work on the tree structure $T$, with root $r$, provided by a DFS on $H$. Then, if $\{v,e\}$ is a vertex-edge cut-pair, $e$ can either be a back-edge, or a tree-edge. Furthemore, in the case that $e$ is a tree-edge, we have the following:

\begin{lemma}
\label{lemma:pos_of_e}
If $\{v,e\}$ is a cut-pair such that $e$ is a tree-edge, then $e$ either lies in $T(v)$ or on the simple tree path $T[v,r]$.
\end{lemma}
\begin{proof}
\label{proof:posOfe}
Suppose that $e$ is neither in $T(v)$ nor on the path $T[v,r]$. Since $e$ is a tree-edge, it has the form $(u,p(u))$, for some vertex $u$. Since $H$ is $2$-vertex-connected, there exists a back-edge $e'=(x,y)$ joining a vertex $x$ from $T(u)$ with a proper ancestor $y$ of $u$. Now, remove $v$ and $e$ from the graph. Since $e$ does not lie on the path $T[v,r]$, $v$ is not a descendant of $u$, and therefore $T(u)$ remains connected. Furthermore, since $e$ is not in $T(v)$, $y$ remains connected with $p(u)$. The existence of $e'$ implies that $u$ remains connected with $p(u)$ - a contradiction.
\end{proof}

Thus we have three distinct cases in total, and we will compute $\mathit{count}(v)$ by counting the cut-pairs $\{v,e\}$ in each case. We will handle these cases separately, by providing a specific algorithm for each one of them, based on some simple observations like Lemma \ref{lemma:pos_of_e}. The linearity of these algorithms will be clear.

Now, we shall begin with the case where $e$ is a back-edge, since this is the easiest to handle. We suppose that all $\mathit{count}(v)$ have been initialized to zero.

\subsection{The case where $e$ is a back-edge}
\label{section:e_is_a_back_edge}

\begin{proposition}
If $\{v,e\}$ is a cut-pair such that $e$ is a back-edge, then $e$ starts from the subtree $\mathit{T}(c)$ of a child $c$ of $v$, ends in a proper ancestor of $v$, and is the only back-edge that starts from $\mathit{T}(c)$ and ends in a proper ancestor of $v$. Conversely, if $e$ is such a back-edge, then $\{v,e\}$ is a cut-pair.
\end{proposition}

\begin{algorithm}[!h]
\caption{\textsf{Calculating $\mathit{up}(c)$ and $\mathit{down}(v,c)$}}
\label{algorithm:UD}
\LinesNumbered
\DontPrintSemicolon

initialize all $\mathit{up}(v)$ and $\mathit{down}(v,c)$ to $0$\;
sort the back-edges $(u,v)$ in increasing order w.r.t. their higher end $u$\;
sort the list of the children of every vertex in increasing order\;
\ForEach{vertex $v$}{
  \lIf{$v$ is not childless}{$c_v \leftarrow $ first child of $v$}
}
\ForEach{back-edge $(u,v)$}{
  $\mathit{up}(u) \leftarrow \mathit{up}(u)+1$\;
  \lWhile{$c_v$ is not an ancestor of $u$}{$c_v \leftarrow $next child of $v$}
  $\mathit{down}(v,c_v) \leftarrow \mathit{down}(v,c_v)+1$\;
}
\end{algorithm}

This immediately suggests an algorithm for counting all such cut-pairs. We only have to count, for every vertex $c$ ($\ne r$ or the child of $r$), the number $\mathit{b\_count}(c)$ $:=$ $\#\{$back-edges that start from $T(c)$ and end in a proper ancestor of $p(c)$$\}$. To do this efficiently, we define, for every vertex $v$, $\mathit{up}(v)$ $:=$ $\#\{$back-edges that start from $v$ and end in an ancestor of $v$$\}$, and, for every child $c$ of $v$ (if it has any), $\mathit{down}(v,c)$ $:=$ $\#\{$back-edges that start from $T(c)$ and end in $v\}$.
Algorithm~\ref{algorithm:UD} computes all $\mathit{up}(v)$ and $\mathit{down}(v,c)$ in linear time.
Now, $\mathit{b\_count}(c)$ can be computed recursively: if $d_1,\dotsc,d_k$ are the children of $c$, then $\mathit{b\_count}(c)=\mathit{up}(c)+\mathit{b\_count}(d_1)+\dotsc+\mathit{b\_count}(d_k)-\mathit{down}(p(c),c)$; and if $c$ is childless, $\mathit{b\_count}(c)=\mathit{up}(c)$. Finally, the number of vertex-edge cut-pairs $\{v,e\}$ where $e$ is a back-edge, equals the number of children $c$ of $v$ that have $b\_count(c)=1$.

\subsection{The case where $e$ is part of the simple tree path $T[v,r]$}
Let $\{v,e\}$ be a vertex-edge cut-pair such that $e$ is part of the simple tree path $T[v,r]$. Then there exists a vertex $u$ which is a proper ancestor of $v$ and such that $e=(u,p(u))$. We observe that all back-edges that start from $\mathit{T}(u)$ and end in a proper ancestor of $u$ must necessarily start from $\mathit{T}(v)$. In other words, $\mathit{M}(u)$ is a descendant of $v$. Here we further distinguish two cases, depending on whether $\mathit{M}(u)$ is a proper descendant of $v$.

\subsubsection{The case $\mathit{M}(u)=v$}
\label{section:Mu=v}
Our algorithm for this case is based on the following observation:

\begin{proposition}
\label{proposition:MuEqv}
Let $c_1,\dotsc,c_k$ be the children of $v$ (if it has any), and let $\{v,(u,\mathit{p}(u))\}$ be a cut-pair such that $u$ is an ancestor of $v$ with $\mathit{M}(u)=v$. Then $u$ does not belong in any set of the form $\mathit{T}[\mathit{high_p}(c_i),\mathit{low}(c_i))$, for $i=1,\dotsc,k$. Conversely, given that $u$ is a proper ancestor of $v$ such that $\mathit{M}(u)=v$, and given also that $u$ does not belong in any set of the form $\mathit{T}[\mathit{high_p}(c_i),\mathit{low}(c_i))$, for $i=1,\dotsc,k$, we may conclude that the pair $\{v,(u,\mathit{p}(u))\}$ is a cut-pair. (See Figure \ref{figure:M(u)=v}.)
\end{proposition}
\begin{proof}
\label{proof:MuEqv}
($\Rightarrow$) Suppose that $u$ belongs to $\mathit{T}[\mathit{high_p}(c),\mathit{low}(c))$, for some child $c$ of $v$. Then $\mathit{high_p}(c)$ is a descendant of $u$, $\mathit{low}(c)$ is an ancestor of $\mathit{p}(c)$, and both of these vertices are proper ancestors of $v$. Now, there exists a back-edge $(x,\mathit{high_p}(c))$, with $x$ in $\mathit{T}(c)$. There also exists a back edge $(x',\mathit{low}(c))$, with $x'$ in $\mathit{T}(c)$. Therefore, it should be clear that the removal of both $v$ and $(u,\mathit{p}(u))$ does not disconnect $u$ from $\mathit{p}(u)$: for, even after this removal, both $u$ and $\mathit{p}(u)$ remain connected with the subtree $\mathit{T}(c)$. This contradicts the fact that $\{v,(u,\mathit{p}(u))\}$ is a cut-pair.\\
($\Leftarrow$) Remove the vertex $v$ and the edge $(u,\mathit{p}(u))$. $\mathit{M}(u)=v$ means that there are no back-edges $(x,y)$ such that $x$ is a descendant of $u$, but not a descendant of $v$, and $y$ is an ancestor of $\mathit{p}(u)$. Therefore, if $u$ remains connected with $\mathit{p}(u)$, they must both be connected with a subtree $\mathit{T}(c)$, for some child $c$ of $v$. Furthermore, if such is the case, there must exist two back-edges, $(x,y)$ and $(x',y')$, such that both $x$ and $x'$ are in $\mathit{T}(c)$, $y$ is proper ancestor of $v$ and a descendant of $u$, and $y'$ is an ancestor of $\mathit{p}(u)$. But this means that $u\leq \mathit{high_p}(c)$ and $\mathit{low}(c)\leq \mathit{p}(u)$. In other words, $u$ is in $\mathit{T}[\mathit{high_p}(c),\mathit{low}(c))$ - a contradiction.
\end{proof}

\begin{figure}[h!]
\begin{center}
\centerline{\includegraphics[trim={0cm 0cm 0cm 8cm}, scale=0.8,clip, width=\textwidth]{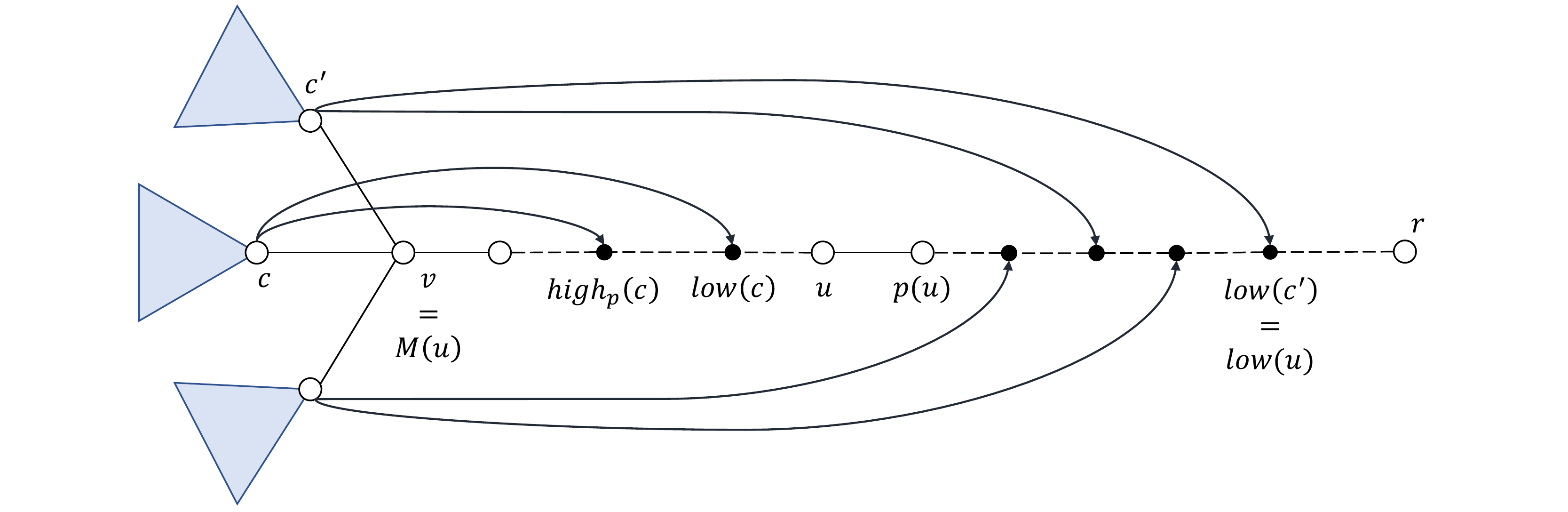}}
\caption{$M(u)=v$ and $v$ forms a cut-pair with $(u,p(u))$.}
\label{figure:M(u)=v}
\end{center}
\end{figure}

Algorithm~\ref{algorithm:Mu=v} describes how we can count, for every vertex $v$, all cut-pairs of the form $\{v,(u,p(u))\}$, where $u$ is a proper ancestor of $v$ with $M(u)=v$.

\begin{algorithm}[!h]
\caption{\textsf{$M(u)=v$}}
\label{algorithm:Mu=v}
\LinesNumbered
\DontPrintSemicolon
calculate all lists $M^{-1}(v)$, for all vertices $v$, and have their elements sorted in decreasing order\;
sort the list of the children of every vertex in decreasing order w.r.t. the $high_p$ value of its elements\;
\ForEach{vertex $v$}{
\lIf{$M^{-1}(v)=\emptyset$}{\textbf{continue}}
  $u \leftarrow $second element of $M^{-1}(v)$ \tcp{the first element of $M^{-1}(v)$ is $v$}
  $c \leftarrow $first child of $v$\;
  $\mathit{min} \leftarrow v$\;
  \While{$u\ne\emptyset$ and $c\ne\emptyset$}{
    $\mathit{min} \leftarrow \mathit{high_p}(c)$\;
    \While{$u\ne \emptyset$ and $u>\mathit{min}$}{
      $\mathit{count}[v] \leftarrow \mathit{count}[v]+1$\;
      $u \leftarrow $next element of $M^{-1}(v)$\;
    }
    $\mathit{min} \leftarrow \mathit{low}(c)$\;
    $c \leftarrow $next child of $v$\;
    \While{$c\ne\emptyset$ and $\mathit{high_p}(c)\geq\mathit{min}$}{
      \lIf{$\mathit{low}(c)<\mathit{min}$}{$\mathit{min} \leftarrow \mathit{low}(c)$}
      $c \leftarrow $next child of $v$\;
    }
    \lWhile{$u\ne \emptyset$ and $u>\mathit{min}$}{$u \leftarrow $next element of $M^{-1}(v)$}
  }
  \While{$u\ne \emptyset$}{
    \lIf{$u\leq\mathit{min}$}{$\mathit{count}[v] \leftarrow \mathit{count}[v]+1$}
    $u \leftarrow $next element of $M^{-1}(v)$\;
  }
}
\end{algorithm}

\begin{theorem}
Algorithm \ref{algorithm:Mu=v} is correct.
\end{theorem}
\begin{proof}
By Proposition \ref{proposition:MuEqv}, we only have to count, for every vertex $v$, the vertices $u$ in $T(v,r)$ that have $M(u)=v$ and are not contained in any set of the form $T[\mathit{high_p}(c),\mathit{low}(c))$, for any child $c$ of $v$. We do this by finding, in a sense, all maximal subsets of $T(v,r)$ of the form $T[x,y]$, which do not meet any set $T[\mathit{high_p}(c),\mathit{low}(c))$, for any child $c$ of $v$, and we count all elements of $M^{-1}(v)\cap T[x,y]$. If $c_1,\dotsc,c_k$ is the list of the children of $v$ sorted in decreasing order w.r.t. their $\mathit{high_p}$ point, then the first such set is $T(v,\mathit{high_p}(c_1))$, the last one is $T[\mathit{low}(c),r)$, where $c$ is a child of $v$ with $\mathit{low}$ minimal among the children of $v$, and all intermediary sets have the form $T[\mathit{low}(c'),\mathit{high_p}(c''))$, for some children $c',c''$ of $v$. If $v$ is childless, we only have to count the elements of $M^{-1}(v)\cap T(v,r)$.
\end{proof}

\subsubsection{The case where $\mathit{M}(u)$ is a proper descendant of $v$}
\label{section:M(u)>v}
In this case, $\mathit{M}(u)$ belongs to $\mathit{T}(c)$, for a child $c$ of $v$, and so we have that $\{\mathit{p}(c),(u,\mathit{p}(u))\}$ is a cut-pair. We base our algorithm for this case on the following observation:

\begin{proposition}
\label{proposition:Mu>v}
Let $\{\mathit{p}(c),(u,\mathit{p}(u))\}$ be a cut-pair, such that $u$ is an ancestor of $\mathit{p}(c)$ and $\mathit{M}(u)$ is in $\mathit{T}(c)$. Then $\mathit{M_p}(c)=\mathit{M}(u)$ and $\mathit{high_p}(c)<u$. Conversely, if $u$ is a proper ancestor of $\mathit{p}(c)$  such that $\mathit{M_p}(c)=\mathit{M}(u)$ and $\mathit{high_p}(c)<u$, then the pair $\{\mathit{p}(c),(u,\mathit{p}(u))\}$ is a cut-pair. (See Figure \ref{figure:M(u)>v}.)
\end{proposition}
\begin{proof}
\label{proof:MuGreaterThanv}
($\Rightarrow$) Let $(x,y)$ be a back-edge such that $x$ is in $\mathit{T}(u)$ and $y$ is a proper ancestor of $u$. Since $\mathit{M}(u)$ is in $\mathit{T}(c)$, $x$ is in $\mathit{T}(c)$. Furthermore, since $u$ is an ancestor of $\mathit{p}(c)$, $y$ is a proper ancestor of $\mathit{p}(c)$. This shows that $\mathit{M_p}(c)$ is an ancestor of $\mathit{M}(u)$. Conversely, let $(x,y)$ be a back-edge such that $x$ is in $\mathit{T}(c)$ and $y$ is a proper ancestor of $\mathit{p}(c)$. Since $c$ is a descendant of $u$, $x$ is in $\mathit{T}(u)$. Furthermore, since $\{\mathit{p}(c),(u,\mathit{p}(u))\}$ is a cut-pair, $y$ must be a proper ancestor of $u$. (Otherwise, we can easily see that, by removing the vertex $\mathit{p}(c)$ and the edge $(u,\mathit{p}(u))$, $u$ remains connected with $\mathit{p}(u)$: for there exists a back-edge connecting a vertex from $T(M(u))$ (which is a subtree of $T(c)$) with $\mathit{low}(u)$, which is an ancestor of $p(u)$, and $(x,y)$ connects $T(c)$ with $T(p(c),u]$.) This means that $x$ is a descendant of $\mathit{M}(u)$, and this shows that $\mathit{M_p}(c)$ is a descendant of $\mathit{M}(u)$. We conclude that $\mathit{M_p}(c)=\mathit{M}(u)$. Finally, since $\{\mathit{p}(c),(u,\mathit{p}(u))\}$ is a cut-pair, it should be clear that $\mathit{high_p}(c)$ must be a proper ancestor of $u$ (the argument is the same as in the parenthesis).\\
($\Leftarrow$) Remove the vertex $\mathit{p}(c)$ and the edge $(u,\mathit{p}(u))$. Now, if there exists a path connecting $u$ with $\mathit{p}(u)$, this path should contain at least one back-edge $(x,y)$ such that either (1) $x$ is in $\mathit{T}(c)$ and $y$ is in $\mathit{T}(\mathit{p}(c),u]$, or (2) $x$ is a descendant of some vertex in $\mathit{T}(\mathit{p}(c),u]$, but not a descendant of $\mathit{p}(c)$, and $y$ is an ancestor of $\mathit{p}(u)$. No back-edge satisfies (1), since $\mathit{high_p}(c)<u$. Furthermore, no back-edge satisfies (2), since $\mathit{M}(u)$ is in $\mathit{T}(c)$. We conclude that $u$ has been disconnected from $\mathit{p}(u)$.
\end{proof}

\begin{algorithm}[!h]
\caption{\textsf{$M(u)>v$}}
\label{algorithm:Mu>v}
\LinesNumbered
\DontPrintSemicolon
calculate all lists $M^{-1}(m)$ and $M_p^{-1}(m)$, for all vertices $m$, and have their elements sorted in decreasing order\;
\ForEach{vertex $m$}{
  $c \leftarrow $ first element of $M_p^{-1}(m)$\;
  $u \leftarrow $ first element of $M^{-1}(m)$\;
  \While{$c \ne \emptyset$ and $u \ne \emptyset$}{
    \lWhile{$u \ne\emptyset$ and $u \geq p(c)$}{$u \leftarrow $ next element of $M^{-1}(m)$}
    \lIf{$u=\emptyset$}{\textbf{break}}
    \If{$\mathit{high_p}(c)<u$}{
      $\mathit{n\_edges} \leftarrow 0$\;
      $\mathit{first} \leftarrow u$\;
      \While{$u \ne\emptyset$ and $\mathit{high_p}(c)<u$}{
        $\mathit{n\_edges} \leftarrow \mathit{n\_edges}+1$\;
        $u \leftarrow $ next element of $M^{-1}(m)$\;
      }
      $\mathit{last} \leftarrow $ predecessor of $u$ in $M^{-1}(m)$\;
      $\mathit{count}[p(c)] \leftarrow \mathit{count}[p(c)]+\mathit{n\_edges}$\;
      $c \leftarrow $ next element of $M_p^{-1}(m)$\;
      \While{$c\ne\emptyset$ and $p(c)>\mathit{last}$}{
        \While{$\mathit{first} \geq p(c)$}{
          $\mathit{n\_edges} \leftarrow \mathit{n\_edges}-1$\;
          $\mathit{first} \leftarrow $ next element of $M^{-1}(m)$\;
        }
        $\mathit{count}[p(c)] \leftarrow \mathit{count}[p(c)]+\mathit{n\_edges}$\;
        $c \leftarrow $ next element of $M_p^{-1}(m)$\;
      }
    }
    \Else{
    $c \leftarrow $ next element of $M_p^{-1}(m)$\;
    }
  }
}
\end{algorithm}

Algorithm~\ref{algorithm:Mu>v} describes how we can compute, for every vertex $v$, the number of cut-pairs of the form $\{v,(u,p(u))\}$, where $u$ is a proper ancestor of $v$ with $M(u)$ in $T(c)$ for a child $c$ of $v$.

\begin{figure}[h!]
\begin{center}
\centerline{\includegraphics[trim={0cm 0cm 0cm 9cm}, scale=0.8,clip, width=\textwidth]{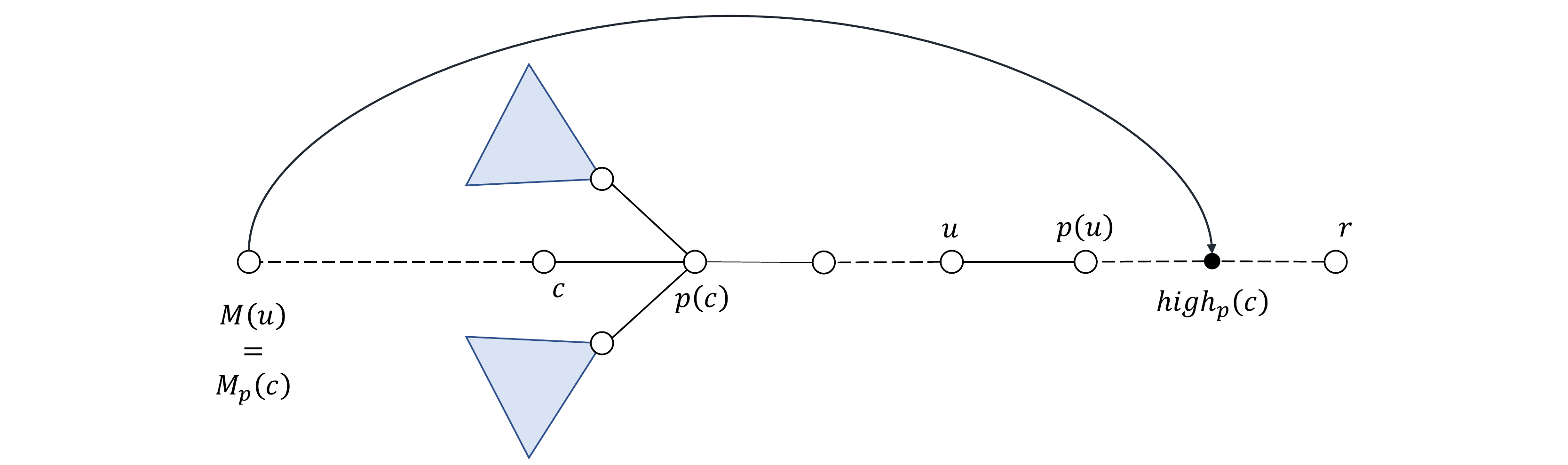}}
\caption{$M(u)$ is in $T(c)$ and $p(c)$ forms a cut-pair with $(u,p(u))$.}
\label{figure:M(u)>v}
\end{center}
\end{figure}

\begin{theorem}
Algorithm \ref{algorithm:Mu>v} is correct.
\end{theorem}
\begin{proof}
According to Proposition \ref{proposition:Mu>v}, for every cut-pair of the form $\{p(c),(u,p(u))\}$ such that $u$ is an ancestor of $p(c)$ with $M(u)$ in $T(c)$, we have $M(u)=M_p(c)$. Therefore, we may search for these cut-pairs by scanning the lists $M^{-1}(m)$ and $M_p^{-1}(m)$, for every vertex $m$. Suppose we have calculated these lists and have their elements sorted in decreasing order. Now, let $c$ be the first element of $M_p^{-1}(m)$ for which there exists a $u$ in $M^{-1}(m)$ such that $u$ is an ancestor of $p(c)$ and the pair $\{p(c),(u,p(u))\}$ is a cut-pair. Proposition \ref{proposition:Mu>v} implies that $\mathit{high_p}(c)<u$. Furthermore, (and, again, as a consequence of the same Proposition), we have that, for every $u'$ in $M^{-1}(m)\cap T(p(c),\mathit{high_p}(c))$, $\{p(c),(u',p(u'))\}$ is a cut-pair, and these are all the elements in $M^{-1}(m)$ for which this is true. Let $U$ denote the collection of these elements. Now, if $c'$ is in $M_p^{-1}(m)\cap T[c,\mathit{high_p}(c))$, then $\mathit{high_p}(c')=\mathit{high_p}(c)$. By Proposition \ref{proposition:Mu>v}, this means that all the elements $u'$ of $M^{-1}(m)$ with the property that $u'$ is a proper ancestor of $p(c')$ and $\{p(c'),(u',p(u'))\}$ is a cut-pair, are precisely the members of $U\cap T(p(c'),\mathit{high_p}(c))$. This explains the counting procedure of Algorithm \ref{algorithm:Mu>v}. Then, after we have updated $\mathit{count}[p(c')]$ for all $c'$ in $M_p^{-1}(m)\cap T[c,\mathit{high_p}(c))$, we repeat the same process for the greatest element $c'$ of $M_p^{-1}(m)$ which is smaller than (i.e. an ancestor of) $\mathit{high_p}(c)$, and has the property that there exists an element $u'$ in $M^{-1}(m)$ such that $u'$ is an ancestor of $p(c')$ and $\{p(c'),(u',p(u'))\}$ is a cut-pair - and keep repeating, until we have traversed $M_p^{-1}(m)$ (or $M^{-1}(m)$) entirely.
\end{proof}

Notice that in Sections \ref{section:e_is_a_back_edge} and \ref{section:Mu=v}, we were able to count specific types of cut-pairs by detecting all of them explicitly.
Here, on the other hand, we count cut-pairs in an indirect manner. Nevertheless, it is easy to see how a minor extension of Algorithm \ref{algorithm:Mu>v} enables us to answer efficiently queries of the form ``report all edges that form a cut-pair (of this type) with a given vertex $v$''.
We only have to store, for every vertex $c$, the lowest ancestor $u$ of $c$, such that $\{p(c),(u,p(u))\}$ is a cut-pair of this type. This is either $\mathit{null}$ or the vertex stored by variable $\mathit{last}$ in Algorithm \ref{algorithm:Mu>v}. Now, to answer a query of the above type, we only have to scan the list of the children of $v$, and if for a child $c$ of $v$ we have that the vertex stored is not $\mathit{null}$, we traverse the (decreasingly sorted) list $M^{-1}(M_p(c))$ backwards, starting from the stored vertex, until we reach an element $u'$ which is not a proper ancestor of $v$. For every vertex $u$ that we encounter (except $u'$), we mark the edge $(u,p(u))$, as one forming a cut-pair with $v$.

Of course, we were bound to perform the counting indirectly at some point: since we claim a linear-time algorithm for the computation of all $\mathit{count}(v)$, we cannot explicitly find all vertex-edge cut-pairs, as their number can be $O(n^2)$. (Consider, for instance, a cycle with $n$ vertices; every vertex $v$ forms precisely $n-2$ vertex-edge cut-pairs, i.e., with all the edges not incident to $v$.) We will perform the counting in an indirect manner again in Section~\ref{section:high(u)<v}. In the next section we find all cut-pairs explicitly.

\subsection{The case where $e$ lies in $\mathit{T}(v)$}
Let $\{v,(u,\mathit{p}(u))\}$ be a cut-pair with $u$ being a descendant of $v$. Then $u$ is a proper descendant of a child $c$ of $v$. Now, we observe that all back-edges that start from $\mathit{T}(u)$ and end in a proper ancestor of $u$ must necessarily end in an ancestor of $\mathit{p}(c)$. In other words, $\mathit{high}(u)\leq v$. Here we distinguish two cases, depending on whether $\mathit{high}(u)$ is a proper ancestor of $v$.

\subsubsection{The case $\mathit{high}(u)=v$}

Our algorithm for this case is based on the following observation:

\begin{proposition}
\label{proposition:high(u)=v}
Let $\{v,(u,\mathit{p}(u))\}$ be a cut-pair such that $v$ is a proper ancestor of $u$ with $\mathit{high}(u)=v$, and let $c$ be the child of $v$ of which $u$ is a descendant. Then, either $(1)$ $\mathit{low}(u)=p(c)$, or $(2)$ $\mathit{low}(u)<p(c)$ and $u\leq \mathit{M_p}(c)$. Conversely, if $c$ is a proper ancestor of $u$ such that $\mathit{high}(u)=\mathit{p}(c)$ and either $(1)$ or $(2)$ holds, then the pair $\{\mathit{p}(c),(u,\mathit{p}(u))\}$ is a cut-pair. (See Figure \ref{figure:high(u)=v}.)
\end{proposition}
\begin{proof}
\label{proof:highuEqv}
($\Rightarrow$) Suppose that $\mathit{low}(u)\neq \mathit{p}(c)$. Then, since $\mathit{low}(u)\leq \mathit{high}(u)$ and $\mathit{high}(u)=\mathit{p}(c)$, we have $\mathit{low}(u)<\mathit{p}(c)$. Furthermore, let $e$ be a back-edge that starts from $\mathit{T}(c)$ and ends in a proper ancestor of $\mathit{p}(c)$. We claim that $e$ starts from $\mathit{T}(u)$. For otherwise, the removal of both $\mathit{p}(c)$ and $(u,\mathit{p}(u))$ would not result in the disconnection of $u$ from $\mathit{p}(u)$. Since, in this case, we could start from $u$, traverse the subtree $\mathit{T}(u)$ until we reach  a vertex from which we can land with a back-edge on $\mathit{low}(u)$, then follow the tree path to the end of $e$ which is a proper ancestor of $\mathit{p}(c)$, and, after we land on the other end of $e$, which is a descendant of a proper ancestor of $u$ which is also a descendant of $c$, we can reach $\mathit{p}(u)$ through a path in $\mathit{T}(c)$. This shows that $\mathit{M_p}(c)$ is in $\mathit{T}(u)$, and therefore we have $u\leq \mathit{M_p}(c)$.\\
($\Leftarrow$) If (1) holds, then all back-edges that start from $\mathit{T}(u)$ and end in a proper ancestor of $u$ end precisely in $\mathit{p}(c)$. In this case, the removal of the pair $\{\mathit{p}(c),(u,\mathit{p}(u))\}$ disconnects the vertices $u$ and $\mathit{p}(u)$. If (2) holds, we claim that $\mathit{M_p}(c)$ is in $\mathit{T}(u)$. Indeed: since $\mathit{low}(u)<\mathit{p}(c)$, there exists a back-edge that starts from $\mathit{T}(u)$ and ends in a proper ancestor of $\mathit{p}(c)$. This implies that $\mathit{M_p}(c)$ is an ancestor of a descendant of $u$. But, since $u\leq \mathit{M_p}(c)$, $\mathit{M_p}(c)$ is not a proper ancestor of $u$. Therefore, it must be a descendant of $u$. Now, since $\mathit{M_p}(c)$ is in $\mathit{T}(u)$ and $\mathit{high}(u)=p(c)$, it is easy to see that the removal of the pair $\{\mathit{p}(c),(u,\mathit{p}(u))\}$ results in the disconnection of $u$ from $\mathit{p}(u)$.
\end{proof}

\begin{figure}[h!]
\begin{center}
\centerline{\includegraphics[trim={0cm 0cm 0cm 7cm}, scale=0.8, clip, width=\textwidth]{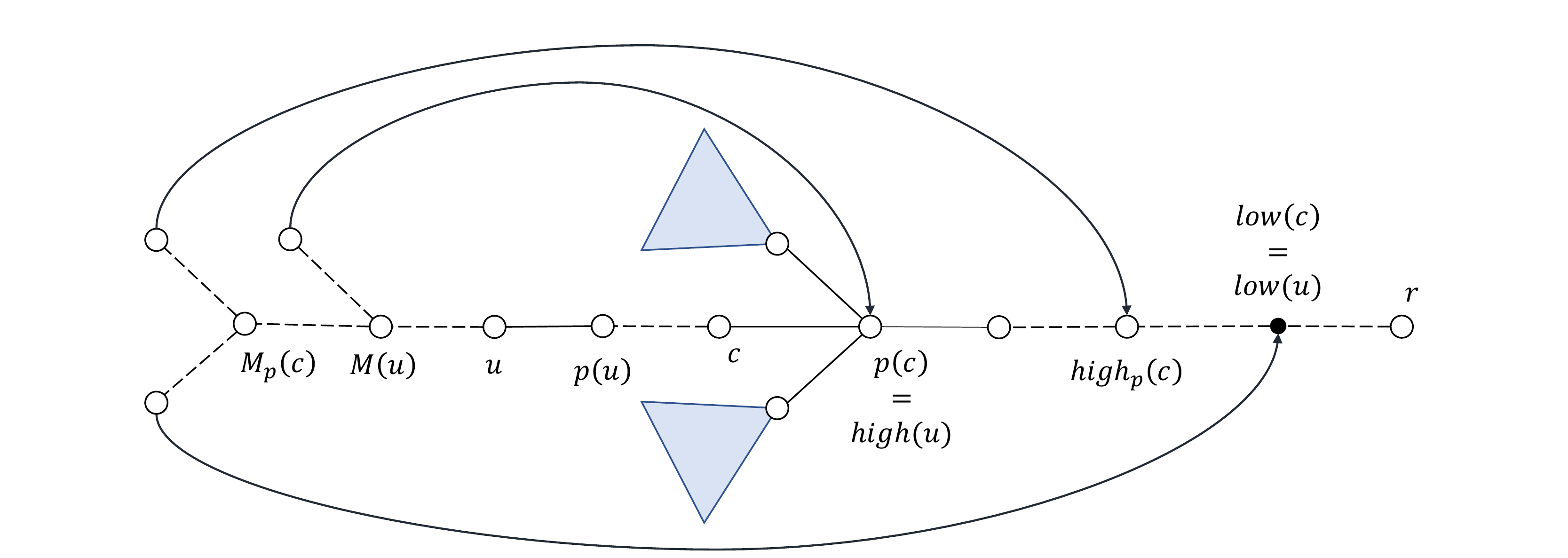}}
\caption{$\mathit{high}(u)=p(c)$ and $p(c)$ forms a cut-pair with $(u,p(u))$. Notice that in this case, where $\mathit{low}(u)<p(c)$, $M_p(c)$ lies in $T(u)$ (and, in fact, is a descendant of $M(u)$).}
\label{figure:high(u)=v}
\end{center}
\end{figure}

It is an immediate application of Proposition~\ref{proposition:high(u)=v} that Algorithm~\ref{algorithm:high(u)=v} correctly computes, for every vertex $v$, the number of cut-pairs $\{v,(u,p(u))\}$ with the property that $u$ is a descendant of $v$ with $\mathit{high}(u)=v$.

\begin{algorithm}[!h]
\caption{\textsf{$\mathit{high}(u)=v$}}
\label{algorithm:high(u)=v}
\LinesNumbered
\DontPrintSemicolon
calculate all lists $\mathit{high}^{-1}(v)$, for all vertices $v$, and have their elements sorted in increasing order\;
sort the list of the children of every vertex in increasing order\;
\ForEach{vertex $v$}{
  $u \leftarrow $ first element of $\mathit{high}^{-1}(v)$\;
  $c \leftarrow $ first child of $v$\;
  \While{$u\ne\emptyset$}{
    \lWhile{$c$ is not an ancestor of $u$}{$c \leftarrow $ next child of $v$}
    \lIf{$\mathit{low}(u)=v$ or ($u\leq\mathit{M_p}(c)$ and $u\neq c$)}
      {$\mathit{count}(v) \leftarrow \mathit{count}(v)+1$}
    $u \leftarrow $ next element of $\mathit{high}^{-1}(v)$\;
  }
}
\end{algorithm}

\subsubsection{The case $\mathit{high}(u)<v$}
\label{section:high(u)<v}

Our algorithm for this case is based on the following observation:

\begin{proposition}
\label{proposition:high(u)<v}
Let $\{\mathit{p}(c),(u,\mathit{p}(u))\}$ be a cut-pair such that $u$ is a descendant of $c$ with $\mathit{high}(u)<\mathit{p}(c)$. Then $\mathit{M}(u)=\mathit{M_p}(c)$. Conversely, if $u$ is a proper descendant of $c$ such that $\mathit{M}(u)=\mathit{M_p}(c)$ and $\mathit{high}(u)<\mathit{p}(c)$, then the pair $\{\mathit{p}(c),(u,\mathit{p}(u))\}$ is a cut-pair. (See Figure \ref{figure:high(u)<v}.)
\end{proposition}
\begin{proof}
\label{proof:highuGreaterThanv}
($\Rightarrow$) Let $(x,y)$ be a back-edge such that $x$ is in $\mathit{T}(u)$ and $y$ is a proper ancestor of $u$. Then, since $u$ is a descendant of $c$, $x$ is in $\mathit{T}(c)$, and, since $\mathit{high}(u)<\mathit{p}(c)$, $y$ is a proper ancestor of $\mathit{p}(c)$. This shows that $\mathit{M_p}(c)$ is an ancestor of $\mathit{M}(u)$. Conversely, let $e$ be a back-edge that starts from $\mathit{T}(c)$ and ends in a proper ancestor of $\mathit{p}(c)$. Then it is easy to see that $e$ must start from $\mathit{T}(u)$ (for otherwise, since $\mathit{high}(u)<\mathit{p}(c)$, the pair $\{\mathit{p}(c),(u,\mathit{p}(u))\}$ would not be a cut-pair, since, by removing it, both $u$ and $p(u)$ would remain connected with $\mathit{high}(u)$). Furthermore, since $\mathit{p}(c)$ is an ancestor of $u$, $e$ ends in a proper ancestor of $u$. This shows that $\mathit{M}(u)$ is an ancestor of $\mathit{M_p}(c)$. Thus we conclude that $\mathit{M}(u)=\mathit{M_p}(c)$.\\
($\Leftarrow$) Remove the vertex $\mathit{p}(c)$ and the edge $(u,\mathit{p}(u))$. Now, if it were possible to reach $\mathit{p}(u)$ from $u$ through a path in the remaining graph, such a path would have to include a back-edge that starts from $\mathit{T}(u)$. But such a back-edge will lead us to a proper ancestor of $\mathit{p}(c)$ (since $\mathit{high}(u)<\mathit{p}(c)$), and then the only way to get back to $\mathit{T}(c)$ (in which we must return, for this is where $\mathit{p}(u)$ lies) is to use a back-edge that starts from $\mathit{T}(c)$ and ends in a proper ancestor of $\mathit{p}(c)$. But such a back-edge must start from $\mathit{T}(u)$ (since $\mathit{M_p}(c)$ lies in $\mathit{T}(u)$). This shows that $\mathit{p}(u)$ cannot be reached from $u$.
\end{proof}

\begin{algorithm}[!h]
\caption{\textsf{$\mathit{high}(u)<v$}}
\label{algorithm:high(u)<v}
\LinesNumbered
\DontPrintSemicolon
calculate all lists $M^{-1}(m)$ and $M_p^{-1}(m)$, for all vertices $m$, and have their elements sorted in decreasing order\;
\ForEach{vertex $m$}{
  $u \leftarrow $ first element of $M^{-1}(m)$\;
  $c \leftarrow $ first element of $M_p^{-1}(m)$\;
  \While{$u\ne\emptyset$ and $c\ne\emptyset$}{
    \lWhile{$c\ne\emptyset$ and $c\geq u$}{$c \leftarrow $ next element of $M_p^{-1}(m)$}
    \lIf{$c=\emptyset$}{\textbf{break}}
    \If{$\mathit{high}(u)<p(c)$}{
      \label{alg_line_1}
      $\mathit{n\_edges} \leftarrow 0$\;
      \label{alg_line_2}
      $h \leftarrow \mathit{high}(u)$\;
      \While{$c\ne\emptyset$ and $h<p(c)$}{
         \label{alg_line_3}
         \While{$u\ne\emptyset$ and $c<u$}{
         \label{alg_line_4}
           $\mathit{n\_edges} \leftarrow \mathit{n\_edges}+1$\;
           $u \leftarrow $ next element of $M^{-1}(m)$\;
         }
         $\mathit{count}[p(c)] \leftarrow \mathit{count}[p(c)]+\mathit{n\_edges}$\;
         $c \leftarrow $ next element of $M_p^{-1}(m)$\;
      }
    }
    \Else{
      $u \leftarrow $ next element of $M^{-1}(m)$\;
    }
  }
}
\end{algorithm}

Algorithm~\ref{algorithm:high(u)<v} describes how we can compute, for every vertex $v$, the number of cut-pairs of the form $\{v,(u,p(u))\}$, where $u$ is a descendant of $v$ with $\mathit{high}(u)<v$.

\begin{figure}[h!]
\begin{center}
\centerline{\includegraphics[trim={0cm 0cm 0cm 8cm}, scale=0.8,clip, width=\textwidth]{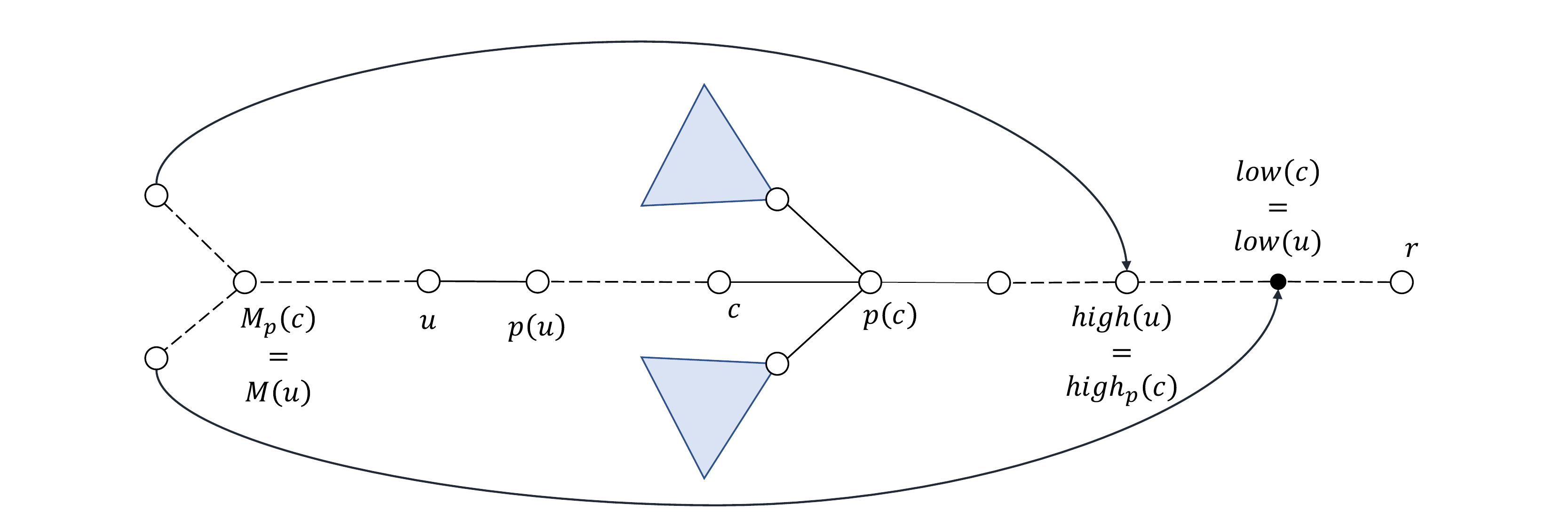}}
\caption{$\mathit{high}(u)<p(c)$ and $p(c)$ forms a cut-pair with $(u,p(u))$.}
\label{figure:high(u)<v}
\end{center}
\end{figure}

\begin{theorem}
Algorithm \ref{algorithm:high(u)<v} is correct.
\end{theorem}
\begin{proof}
According to Proposition \ref{proposition:high(u)<v}, for every cut-pair of the form $\{p(c),(u,p(u))\}$ with $c$ an ancestor of $u$ and $\mathit{high}(u)<p(c)$, we have $M(u)=\mathit{M_p}(c)$. Therefore, in order to count all these pairs, it is sufficient to focus our attention, for every vertex $m$, on the lists $M^{-1}(m)$ and $M_p^{-1}(m)$. Now, fix a vertex $m$, and let $U(c)$, for a vertex $c$ in $M_p^{-1}(m)$, denote the (possibly empty) set of all $u$ in $M^{-1}(m)$ with the property that $u$ is a (proper) descendant of $c$ such that $\{p(c),(u,p(u))\}$ is a cut-pair. Let $c$ be a vertex in $M_p^{-1}(m)$ such that $U(c)$ is not empty, and let $u$ be the greatest element in $U(c)$. Proposition \ref{proposition:high(u)<v} implies that $\mathit{high}(u)<p(c)$. Now let $c'$ be a vertex in $M_p^{-1}(m)$ such that $c'\leq c$ and $\mathit{high}(u)<p(c')$. Since every $u'$ in $M^{-1}(m)\cap T[u,\mathit{high}(u))$ has $\mathit{high}(u')=\mathit{high}(u)$, Proposition \ref{proposition:high(u)<v} implies that every $u'$ in $U(c)$ is also in $U(c')$ (since such a $u'$ is also a proper descendant of $c'$). Furthermore, we claim that no $u'$ strictly greater than $u$ can be in $U(c')$: since $u$ is the greatest element in $M^{-1}(m)$ that is a descendant of $c$ such that $\{p(c),(u,p(u))\}$ is a cut-pair, for every $u'$ in $M^{-1}(m)$ strictly greater than $u$ we must have $\mathit{high}(u')\geq p(c)$ (by Proposition \ref{proposition:high(u)<v}), and therefore $\mathit{high}(u')\geq p(c')$; our claim follows now from Proposition \ref{proposition:high(u)<v}. We conclude that $U(c')=U(c)\cup( T[c,c') \cap M^{-1}(m))$. Therefore, if $\#U(c)$ is known, in order to find $\#U(c')$ it is sufficient to find the elements of $U(c')$ in $M^{-1}(m)\cap T[c,c')$ - call their collection $C$ - and then $\#U(c')=\#U(c)+\#C$. This explains the counting procedure in Algorithm \ref{algorithm:high(u)<v}.
Now, suppose that we have all the lists $M^{-1}(m)$ and $M_p^{-1}(m)$ computed and their elements sorted in decreasing order. Algorithm \ref{algorithm:high(u)<v} works by finding the first $u$ in $M^{-1}(m)$ with the property that there exists a $c$ in $M_p^{-1}(m)$ such that $c$ is an ancestor of $u$ and $\{p(c),(u,p(u))\}$ is a cut-pair. Now, thanks to what we said above, we can easily calculate $\#U(c)$, for every $c$ in $M_p^{-1}(m)$ that has $\mathit{high}(u)<p(c)$. Then, after we have properly updated all $\mathit{count}[p(c)]$, for every such $c$, we only have to repeat the same process for the greatest element $u'$ in $M^{-1}(m)$ which is lower than $\mathit{high}(u)$ and such that there exists a $c$ in $M_p^{-1}(m)$ such that $c$ is an ancestor of $u'$ and $\{p(c),(u',p(u'))\}$ is a cut-pair - and keep repeating, until we reach the end of $M^{-1}(m)$ (or $M_p^{-1}(m)$).
\end{proof}

As in Section~\ref{section:M(u)>v}, a minor extension of Algorithm \ref{algorithm:high(u)<v} enables us to efficiently answer queries of the form
``report all edges that form a cut-pair (of this type) with a given vertex $v$''.
We only have to store, for every vertex $c$, a vertex $\mathit{max}[c]$, the greatest descendant $u$ of $c$ with the property that $\{p(c),(u,p(u))\}$ is a cut-pair of this type (or $\mathit{null}$ if no such vertex $u$ exists).
To do this, between lines \ref{alg_line_1} and \ref{alg_line_2} we keep $u$ stored in a variable $\mathit{temp}$, and between lines \ref{alg_line_3} and \ref{alg_line_4} we set $\mathit{max}[c] \leftarrow \mathit{temp}$. Now, to answer a query of the above type, we scan the list of the children of $v$, and for every child $c$ of $v$ with $\mathit{max}[c]\ne\mathit{null}$, we traverse the (decreasingly sorted) list $M^{-1}(M_p(c))$ from $\mathit{max}[c]$ until we reach an element $u'$ which is not a proper descendant of $c$. For every vertex $u$ that we encounter (except $u'$), we mark the edge $(u,p(u))$, as one forming a cut-pair with $v$.
\\

Finally, let us briefly explain why Algorithms \ref{algorithm:Mu=v}, \ref{algorithm:Mu>v}, \ref{algorithm:high(u)=v}, and \ref{algorithm:high(u)<v}, run in linear time. All the required sorted lists can be computed in linear time by bucket sorting. For example, we can sort the list of children of $v$, for all vertices $v$, in increasing order w.r.t. the $\mathit{high_p}$ values (as needed in Algorithm \ref{algorithm:Mu=v}), as follows.
First, we initialize all lists $\mathit{high_p^{-1}}(x)$ to $\emptyset$. Then, for every vertex $c$ ($\neq r$ or the child of $r$), we insert into the list $\mathit{high_p^{-1}}(\mathit{high_p}(c))$ the element $c$. Now we initialize the list of children of every vertex $v$ to $\emptyset$. We process all vertices in increasing order, and for every vertex $x$ we do the following: we traverse the list $\mathit{high_p^{-1}}(x)$, and for every $c$ in $\mathit{high_p^{-1}}(x)$ we insert into the list of children of $p(c)$ the element $c$. The computation of all $M^{-1}(m)$, $M_p^{-1}(m)$, $\mathit{high}^{-1}(x)$, etc., is performed in a similar manner. Now, the key observation to see why the main part of Algorithms \ref{algorithm:Mu=v}, \ref{algorithm:Mu>v}, \ref{algorithm:high(u)=v}, and \ref{algorithm:high(u)<v} runs in linear time, is that the final step in every \textbf{while} loop in them always moves forward to the next element of the list under consideration (and there is no backward movement).

\clearpage

\appendix

\section{Computing vertex-edge cut-pairs via SPQR trees}
\label{section:SPQR}

An SPQR tree~\cite{SPQR:triconnected,SPQR:planarity} $\mathcal{T}$ for a biconnected graph $H$ represents the triconnected components of $H$. Each node $\alpha \in \mathcal{T}$  is associated with an undirected graph or multigraph $H_\alpha$. The node $\alpha$, and the graph $H_\alpha$ associated with it, has one of the following types:
\begin{itemize}
\item If $\alpha$ is an $S$-node, then $H_\alpha$ is a cycle graph with three or more vertices and edges.
\item If $\alpha$ is a $P$-node, then  $H_\alpha$ is a multigraph with two vertices and at lest $3$ parallel edges.
\item If $\alpha$ is a $Q$-node, then $H_\alpha$ is a single real edge.
\item If $\alpha$ is an $R$-node, then $H_\alpha$ is a simple triconnected graph.
\end{itemize}
Each edge $\{\alpha,\beta\}$ between two nodes of the SPQR tree is associated with two \emph{virtual edges}, where one is an edge in $H_\alpha$ and the other is an edge in $H_\beta$.
If $\{u,v\}$ is a separation pair in $H$, then one of the following cases applies (see, e.g., \cite{wiki:SPQR}):
\begin{itemize}
\item[(a)] $u$ and $v$ are the endpoints of a virtual edge in the graph $H_\alpha$ associated with an $R$-node $\alpha$ of $\mathcal{T}$.
\item[(b)] $u$ and $v$ are vertices in the graph $H_\alpha$ associated with a $P$-node $\alpha$ of $\mathcal{T}$.
\item[(c)] $u$ and $v$ are vertices in the graph $H_\alpha$ associated with an $S$-node $\alpha$ of $\mathcal{T}$, such that either $u$ and $v$ are not adjacent, or the edge $\{u,v\}$ is virtual.
\end{itemize}

In case (c), if $\{u,v\}$ is a virtual edge, then $u$ and $v$ also belong to a $P$-node or an $R$-node. If $u$ and $v$ are not adjacent then $H \setminus \{u,v\}$ consists of two components that are represented by two paths of the cycle graph $H_\alpha$ associated with the $S$-node $\alpha$ and with the SPQR tree nodes attached to those two paths.
Let $e=\{x,y\}$ be an edge of $H$ such that $\{v,e\}$ is a vertex-edge cut pair of $H$.
Then, $\mathcal{T}$ must contain an $S$-node $\alpha$ such that $v$, $x$ and $y$ are vertices of $H_\alpha$ and $\{u,v\}$ is not a virtual edge.
This observation implies that we can identify (and count) all vertex-edge cut pairs of $H$ by using $\mathcal{T}$.
Gutwenger and P. Mutzel~\cite{SPQR:GM} showed that an SPQR tree can be constructed in linear time, by extending the triconnected components algorithm of Hopcroft and Tarjan.
Hence, this gives a linear-time algorithm to our problem.

\section{Counting all cut-pairs in linear time}
\label{section:finding_cut_pairs}
Let $H=(V,E)$ be a $2$-edge-connected undirected graph. For every $e$ in $E$, we define $\mathit{count}(e)$ $:=$ $\#\{e' \in E \mid \{e,e'\}$ is a cut-pair$\}$. We will find all edges that belong to a cut-pair of $H$ by computing all $\mathit{count}(e)$. The parameter $\mathit{count}(e)$ is also useful for counting TSCCs, as {{}} suggested by the following Lemma.

\begin{lemma}
\label{lemma:counting_tsccs_edge}
Let $G$ be a twinless strongly connected digraph, and let $e$ be a twinless strong bridge of $G$ which is not a strong bridge. Then $\mathit{count}(\tilde{e})+1$, where $\tilde{e}$ is the edge of $G^u$ corresponding to $e$, is the number of TSCCs of $G\setminus{e}$.
\end{lemma}
\begin{proof}
Since $e$ is not a strong bridge, $G\setminus{e}$ is strongly connected. By Theorem~\ref{theorem:TwinlessCharacterization}, the TSCCs of a strongly connected digraph correspond to the $2$-edge-connected components of its underlying undirected graph. Now, the number of the $2$-edge-connected components of $(G\setminus{e})^u = G^u \setminus \tilde{e}$ equals the number of its bridges + $1$ (this is due to the tree structure of the $2$-edge-connected components of a connected graph). By definition, this number is equal to $\mathit{count}(\tilde{e})$.
\end{proof}

To compute all $\mathit{count}(e)$, we will work on the tree structure $T$, with root $r$, provided by a DFS on $H$. Then, if $\{e,e'\}$ is a cut-pair of edges, either one of them is a back-edge and the other one is a tree-edge, or both of them are tree-edges. (It cannot be the case that both of them are back-edges, since their removal would not disconnect the graph.) Furthermore, in the case that both of them are tree-edges, we have the following:

\begin{lemma}
\label{lemma:cut_pairs_on_line}
If $\{e,e'\}$ is a cut-pair such that both $e$ and $e'$ are tree-edges, then they both lie on the simple tree path $T[u,r]$, for some vertex $u$.
\end{lemma}
\begin{proof}
Since both $e$ and $e'$ are tree-edges, there exist vertices $u$ and $v$, such that $e=(u,p(u))$ and $e'=(v,p(v))$. Since the graph is $2$-edge-connected, $u$ is distinct from $v$, and let's assume, without loss of generality, that $u>v$. Now, suppose that $e$ and $e'$ are not part of the simple tree path $T[u,r]$. Then $v$ is not an ancestor of $u$. Furthermore, since $u>v$, $v$ is not a descendant of $u$ either. Now, remove $\{e,e'\}$ from the graph. We note three facts. First, since $v$ is not a descendant of $u$, $T(u)$ remains connected. Second, since the graph is $2$-edge-connected, there exists a back-edge connecting a vertex from $T(u)$ with a proper ancestor $x$ of $u$. Third, since $v$ is not an ancestor of $u$, the vertices on the simple tree path $T[p(u),x]$ remain connected. These three facts imply that $u$ remains connected with $p(u)$, and therefore $\{e,e'\}$ is not a cut-pair. A contradiction.
\end{proof}

Thus we have two distinct cases to consider, and we will compute $\mathit{count}(e)$ by counting the number of cut-pairs $\{e,e'\}$ in each case. We will handle these cases separately, by providing a specific algorithm for each one of them. We shall begin with the case that one of those edges is a back-edge, since this is the easiest to handle. We suppose that all $\mathit{count}(e)$ have been initialized to zero.

\subsection{The case back-edge - tree-edge}
\label{subsection:back_edge_tree_edge}
Our algorithm for this case is based on the following observation:
\begin{proposition}
\label{proposition:back_edge}
Let $e$ be a back-edge and $u$ a vertex distinct from $r$. The pair $\{e,(u,p(u))\}$ is a cut-pair if and only if $e$ starts from $T(u)$, ends in a proper ancestor of $u$, and is the only back-edge with this property.
\end{proposition}
\begin{proof}
{{}}
($\Rightarrow$) Since the graph is $2$-edge-connected, there exists at least one back-edge $e''$ with the property that $e''$ connects a descendant of $u$ with a proper ancestor of $u$. Supposing $e''\neq e$, we see that $\{e,(u,p(u))\}$ cannot be a cut-pair: for in this case, by removing $\{e,(u,p(u))\}$, $u$ remains connected with $p(u)$. A contradiction.\\
($\Leftarrow$) Remove the edge $(u,p(u))$ from the graph. Now every path that connects $u$ with $p(u)$ must necessarily use a back-edge connecting a vertex from $T(u)$ with a proper ancestor of $u$. If $e$ is the only back-edge with this property, its removal leaves $u$ and $p(u)$ in different connected components, and thus $\{e, (u,p(u))\}$ is a cut-pair.
\end{proof}

This implies that, for every vertex $u$, there exists at most one cut-pair of the form $\{e,(u,p(u))\}$, where $e$ is a back-edge, and it immediately suggests an algorithm for determining whether such a cut-pair exists. We only have to count, for every vertex $u$ ($\neq r$), the number $\mathit{b\_count}(u)$ $:=$ $\#\{$back-edges that start from $T(u)$ and end in a proper ancestor of $u$$\}$. To do this, we define, for every vertex $u$, $\mathit{up}(u)$ $:=$ $\#\{$back-edges that start from $u$ and end in an ancestor of $u$$\}$, and $\mathit{down}(u)$ $:=$ $\#\{$back-edges that start from $T(u)$ and end in $u$$\}$. 
All $\mathit{up}(u)$ and $\mathit{down}(u)$ can be computed easily in linear time, during the DFS. $\mathit{b\_count}(u)$ can be computed recursively: if $c_1,\dotsc,c_k$ are the children of $u$, then $\mathit{b\_count}(u)=\mathit{up}(u)+\mathit{b\_count}(c_1)+\dotsc+\mathit{b\_count}(c_k)-\mathit{down}(u)$; if $u$ is childless, $\mathit{b\_count}(u)=\mathit{up}(u)$.
Now, for every vertex $u$ that has $\mathit{b\_count}(u)=1$, we set $\mathit{count}[(u,p(u)] \leftarrow \mathit{count}[(u,p(u)]+1$. Furthermore, in this case there exists only one back-edge $(x,y)$ such that $x$ is a descendant of $u$ and $y$ is a proper ancestor of $u$, and so we have $x=M(u)$ and $y=\mathit{low}(u)$. Thus we also set $\mathit{count}[(M(u),\mathit{low}(u))] \leftarrow \mathit{count}[(M(u),\mathit{low}(u))]+1$.

\subsection{The case where both edges are tree-edges}
Our algorithm for this case is based on the following observation:

\begin{proposition}
\label{proposition:tree_edges_cut}
Let $u$, $v$ be two vertices such that $v$ is an ancestor of $u$. Then $\{(u,p(u)),(v,p(v))\}$ is a cut-pair if and only if $v$ is a proper ancestor of $u$ with $M(u)=M(v)$ and $\mathit{high}(u)<v$.
\end{proposition}
\begin{proof}
($\Rightarrow$) Since the graph is $2$-edge-connected, the removal of one edge is not sufficient to disconnect the graph, and therefore $v$ must be a \textit{proper} ancestor of $u$. Now, let $(x,y)$ be a back-edge such that $x$ is a descendant of $u$ and $y$ is a proper ancestor of $u$. Since $u$ is a descendant of $v$, $x$ is also a descendant of $v$. Furthermore, we notice that, since $\{(u,p(u)),(v,p(v))\}$ is a cut-pair, $y$ is a proper ancestor of $v$. (For otherwise, by removing $(u,p(u))$ and $(v,p(v))$, $T(u)$ remains connected (since $v$ is an ancestor of $u$), the vertices in the simple tree path $T[p(u),y]$ remain connected (since $v$ is an ancestor of $y$), and, therefore, the existence of the back-edge $(x,y)$ implies that $u$ remains connected with $p(u)$.) This shows that $M(v)$ is an ancestor of $M(u)$, and $\mathit{high}(u)$ is a proper ancestor of $v$. Conversely, let $(x,y)$ be a back-edge such that $x$ is a descendant of $v$ and $y$ is a proper ancestor of $v$. We observe that, since $\{(u,p(u)),(v,p(v))\}$ is a cut-pair, $x$ must be a descendant of $u$. (For otherwise, by removing $(u,p(u))$ and $(v,p(v))$, the simple tree paths $T[x,v]$ and $T[p(v),y]$ have not been affected (since $u$ is a descendant of $x$), and, therefore, the existence of the back-edge $(x,y)$ implies that $v$ remains connected with $p(v)$.) Furthermore, since $v$ is an ancestor of $u$, $y$ is a proper ancestor of $u$. This shows that $M(u)$ is an ancestor of $M(v)$. We conclude that $M(u)=M(v)$.\\
($\Leftarrow$) Remove the edges $(u,p(u))$ and $(v,p(v))$. Now, if there exists a path connecting $u$ with $p(u)$, this path must use a back-edge $(x,y)$ such that either (1) $x$ is in $T(u)$ and $y$ in $T[p(u),v]$, or (2) $x$ is a descendant of a vertex in $T[p(u),v]$, but not a descendant of $u$, and $y$ is a proper ancestor of $v$. A back-edge satisfying (1) does not exist, since $\mathit{high}(u)<v$. A back-edge $(x,y)$ satisfying (2) does not exist, since $M(u)=M(v)$, and therefore, if $x$ is a descendant of $v$ and $y$ a proper ancestor of $v$, then $x$ is a descendant of $M(u)$, and therefore a descendant of $u$. We conclude that the pair $\{(u,p(u)),(v,p(v))\}$ is a cut-pair.
\end{proof}

Algorithm \ref{algorithm:tree_edges_cut} describes how we can compute, for every vertex $u$, the number of cut-pairs of the form $\{(u,p(u)),(v,p(v))\}$. To prove correctness, we will need the following:

\begin{lemma}
\label{lemma:same_M(v)}
Let $u$, $v$ be two vertices such that $M(u)=M(v)$, $v$ is an ancestor of $u$, and $\mathit{high}(u)<v$. Then $\mathit{high}(u)=\mathit{high}(v)$.
\end{lemma}
\begin{proof}
Let $(x,y)$ be a back-edge such that $x$ is a descendant of $u$ and $y$ is a proper ancestor of $u$. Then, since $u$ is a descendant of $v$, $x$ is also a descendant of $v$. Furthermore, $\mathit{high}(u)<v$ implies that $y$ is a proper ancestor of $v$. This shows that $y\leq\mathit{high}(v)$, and therefore, since $\mathit{high}(u)\leq y$, we have $\mathit{high}(u)\leq\mathit{high}(v)$. Conversely, let $(x,y)$ be a back-edge such that $x$ is a descendant of $v$ and $y$ is a proper ancestor of $v$. Then, since $M(v)=M(u)$, $M(v)$ is a descendant of $u$, and therefore $x$ is also a descendant of $u$. Furthermore, since $v$ is an ancestor of $u$, $y$ is a proper ancestor of $u$. This shows that $y\leq \mathit{high}(u)$, and therefore, since $\mathit{high}(v)\leq y$, we have $\mathit{high}(v)\leq\mathit{high}(v)$. We conclude that $\mathit{high}(u)=\mathit{high}(v)$.
\end{proof}

Now, we let $S(u)$, for every vertex $u\neq r$, denote the set $\{u\} \cup \{ v  \mid  \{(u,p(u)),(v,p(v))\}$ is a cut-pair $\}$. Then we have the following:

\begin{proposition}
\label{proposition:tree_edges_cut_2}
For every $v \in S(u)$ we have $S(v)=S(u)$.
\end{proposition}
\begin{remark}
In other words, this Proposition says that the binary relation ``$e$ forms a cut-pair with $e'$'', defined on the set of tree-edges, is transitive.
\end{remark}
\begin{proof}
Let $u$ be a vertex. We will show that all vertices in $S(u)$ have the same $\mathit{high}$ point. So let $v$ be a member of $S(u)$. By Lemma \ref{lemma:cut_pairs_on_line}, we have that $v$ is either an ancestor of $u$ or a descendant of $u$. Suppose that $v$ is an ancestor of $u$. Then, Proposition \ref{proposition:tree_edges_cut} implies that $\mathit{high}(u)<v$. By the same Proposition, we also have $M(u)=M(v)$. By Lemma \ref{lemma:same_M(v)}, these two facts imply that $\mathit{high}(u)=\mathit{high}(v)$. Now, if $v$ is a descendant of $u$, the same argument (with a reversal of the roles of $u$ and $v$) shows that $\mathit{high}(u)=\mathit{high}(v)$.

Now let $v$ be a member of $S(u)$ and $w$ a member of $S(v)$. Since $\mathit{high}(w)=\mathit{high}(v)$ and $\mathit{high}(v)=\mathit{high}(u)$, we conclude that $\mathit{high}(w)=\mathit{high}(u)$, and therefore (by the definition of $\mathit{high}$) $\mathit{high}(w)<u$ and $\mathit{high}(u)<w$. By Proposition \ref{proposition:tree_edges_cut}, we have $M(w)=M(v)$ and $M(v)=M(u)$, and thus we conclude that $M(w)=M(u)$. Now, since $u$ and $v$ are related as ancestor-descendant, and $v$ and $w$ also have the same relation, we conclude that $u$ and $w$ are also related as ancestor-descendant. Now, it does not matter whether $u$ is an ancestor or a descendant of $w$: since $M(u)=M(w)$ and $\mathit{high}(u)<w$ and $\mathit{high}(w)<u$, either one is a proper ancestor of the other, and, therefore, by Proposition \ref{proposition:tree_edges_cut}, we have that $w$ is in $S(u)$, or $w=u$, and therefore $w$ is in $S(u)$ (by definition). Thus we have $S(v)\subseteq S(u)$. Since $u \in S(v)$ (by definition), the same argument shows that $S(u) \subseteq S(v)$, and thus we conclude that $S(v)=S(u)$. 
\end{proof}

\begin{theorem}
Algorithm \ref{algorithm:tree_edges_cut} is correct.
\end{theorem}
\begin{proof}
According to Proposition \ref{proposition:tree_edges_cut}, all cut-pairs of the form $\{(u,p(u)),(v,p(v))\}$ have $M(u)=M(v)$. Therefore, in order to count all these cut-pairs, it is sufficient to focus our attention on the lists $M^{-1}(m)$, for all vertices $m$ ($\neq r$), to find therein vertices $u$ and $v$ such that $\{(u,p(u)),(v,p(v))\}$ is a cut-pair. Now, suppose that we have all these lists computed and their elements sorted in decreasing order, and let $m$ be a vertex. Let $u$ be an element of $M^{-1}(m)$ which is maximal in $S(u)$ (i.e. if there exists a $v$ such that $\{(u,p(u)),(v,p(v))\}$ is a cut-pair, then $v$ is a proper ancestor of $u$). Then, by Proposition \ref{proposition:tree_edges_cut}, we have $S(u)=M^{-1}(m)\cap T[u,\mathit{high}(u))$. Furthermore, by Proposition \ref{proposition:tree_edges_cut_2}, we have that, for every $v$ in $S(u)$, $S(v)=S(u)$. This explains why Algorithm \ref{algorithm:tree_edges_cut} works. We start with the first element $u$ of $M^{-1}(m)$, and we traverse the list $M^{-1}(m)$ until we reach a vertex $v$ such that $v\leq\mathit{high}(u)$ (or until we run out of elements). While doing that, we keep a counter $\mathit{n\_edges}$ of the elements in $M^{-1}(m)\cap T(u,\mathit{high}(u))$. Then we traverse the segment $M^{-1}(m)\cap T[u,\mathit{high}(u))$ of the list $M^{-1}(m)$ again, and, for every $w$ in $M^{-1}(m)\cap T[u,\mathit{high}(u))$, we set $\mathit{count}[(w,p(w))] := \mathit{count}[(w,p(w))]+\mathit{n\_edges}$. Then we repeat the same process from $v$, until we reach the end of $M^{-1}(m)$.
\end{proof}

\begin{algorithm}[!h]
\caption{\textsf{counting cut-pairs of tree-edges}}
\label{algorithm:tree_edges_cut}
\LinesNumbered
\DontPrintSemicolon
calculate all lists $M^{-1}(m)$, for all vertices $m$, and have their elements sorted in decreasing order\;
\ForEach{vertex $m$}{
$u \leftarrow $ first element of $M^{-1}(m)$\;
\While{$u\neq\emptyset$}{
  $v \leftarrow $ successor of $u$ in $M^{-1}(m)$\;
  $\mathit{n\_edges} \leftarrow 0$\;
  \While{$v\neq\emptyset$ and $\mathit{high}(u)<v$}{
    $\mathit{n\_edges} \leftarrow \mathit{n\_edges}+1$\;
    $v \leftarrow $ next element of $M^{-1}(m)$\;
   }
  $v \leftarrow u$\;
  \While{$v\neq\emptyset$ and $\mathit{high}(u)<v$}{
    $\mathit{count}[(v,p(v))] \leftarrow \mathit{count}[(v,p(v))]+\mathit{n\_edges}$\;
    $v \leftarrow $ next element of $M^{-1}(m)$\;
  }
  $u \leftarrow v$\;
}
}
\end{algorithm}

In conclusion, we note that the method we have proposed to count all cut-pairs allows us to preprocess the graph in linear time so that we can answer queries of the form ``report all edges that form a cut-pair with $e$" in time analogous to the length of the answer. Let us explain how this works. Suppose, first, that $e$ is a back-edge. Then all the edges with which $e$ forms a cut-pair (if such edges exist) must be tree-edges. By Proposition \ref{proposition:back_edge}, if $e$ forms a cut-pair with a tree-edge $(u,p(u))$, for some vertex $u$, then there is no other back-edge $e'$ such that $\{e',(u,p(u))\}$ is a cut-pair. Furthermore, in this case we have $e=(M(u),\mathit{low}(u))$, and $\mathit{low}(u)=\mathit{low}(M(u))$. Thus, we may store in a list $\mathit{b\_cut}[v]$, for every vertex $v$ ($\neq r$), all vertices $u$ such that $M(u)=v$ and $\{(v,\mathit{low}(v)),(u,p(u))\}$ is a cut-pair. To fill these lists appropriately, we process each vertex $u\neq r$: if $\mathit{b\_count}(u)$ (as defined in the beginning of this Appendix) is equal to $1$, we insert in $\mathit{b\_cut}[M(u)]$ the element $u$. Now, to answer a query of the form ``report all edges that form a cut-pair with $(x,y)$", where $(x,y)$ is a back-edge, we check whether $y=\mathit{low}(x)$, in which case we return all edges of the form $(u,p(u))$, for every $u$ in $\mathit{b\_cut}[x]$ (if $\mathit{b\_cut}[x]$ is not empty). 
Now, suppose that $(u,p(u))$ is a tree-edge. In this case we cannot store in a list all the edges $e$ such that $\{e,(u,p(u))\}$ is a cut-pair (for otherwise we may violate the time bound in the preprocessing phase). Instead, we may store in a variable $\mathit{Max}[u]$, for every vertex $u\neq r$, the maximal element of $S(u)$. To do this, we only have to insert between lines 12 and 13 of Algorithm \ref{algorithm:tree_edges_cut} the assignment ``$\mathit{Max}[v] \leftarrow u $''. Then, to answer a query of the form ``report all edges that form a cut-pair with $(u,p(u))$", we first have to check whether $\mathit{b\_count}(u)=1$, in which case we return (as part of the output) the edge $(M(u),\mathit{low}(u))$. Then, we traverse the (decreasingly ordered) list $M^{-1}(M(u))$ from $\mathit{Max}[u]$ until we reach an element $w$ such that $\mathit{high}(w)\leq \mathit{high}(\mathit{Max}[u])$, or the end of the list. For every such element $v$ that we encounter - excluding $u$ and $w$ - we return $(v,p(v))$.
\end{document}